\documentclass{article}
\usepackage{amsmath,bm}
\usepackage{float}
\usepackage{todonotes}
\usepackage{pgfplots}
\pgfplotsset{compat=1.18} 
\usepackage{afterpage}

\usepackage{cite}
\usepackage{xcolor}
\usepackage{xstring}
\usepackage{caption}
\usepackage[english]{babel}
\usepackage[utf8x]{inputenc}
\usepackage{amsfonts}
\usepackage{amssymb}
\usepackage{amsthm}
\usepackage{theoremref}
\usepackage{enumitem}
\usepackage{subfig}
\usepackage{tikz}
\usepackage{tkz-euclide}
\usetikzlibrary{decorations.pathreplacing, patterns, positioning}
\usetikzlibrary{shapes.geometric, arrows.meta, positioning, fit, backgrounds, calc}
\usepackage{tikz-cd}
\usetikzlibrary{positioning} 
\usepackage[export]{adjustbox}
\usepackage{hyperref}
\hypersetup{
  colorlinks   = true, 
  urlcolor     = blue, 
  linkcolor    = blue, 
  citecolor   =  green 
}
\usepackage[section]{placeins} 
\usepackage{ifthen}
\usepackage{authblk}
\usepackage{mathtools,bm}
\usepackage[linesnumbered,ruled,vlined]{algorithm2e}
\SetKwInput{KwInput}{Input}                
\SetKwInput{KwOutput}{Output}              

\usepackage{pdfsync}
\usepackage{geometry}
\usepackage{graphicx}
\usepackage[normalem]{ulem}
\usepackage{booktabs}
\usepackage{etoolbox}
\usepackage{tabulary,capt-of}
\usepackage{xfrac}
\usepackage{multicol}\columnseprule 0.4pt\raggedcolumns

\def\BibTeX{{\rm B\kern-.05em{\sc i\kern-.025em b}\kern-.08em
    T\kern-.1667em\lower.7ex\hbox{E}\kern-.125emX}}

\definecolor{shadecolor}{RGB}{200,200,200}
\definecolor{darkgreen}{RGB}{20,130,30}
\definecolor{bblack}{RGB}{35,55,59}
\definecolor{ForestGreen}{RGB}{34,139,34}

\newtheorem{theorem}{Theorem}
\newtheorem{lemma}{Lemma}
\newtheorem{proposition}{Proposition}
\newtheorem{corollary}{Corollary}

\newtheorem{remark}{Remark}
\newtheorem{problem}{Problem}
\newtheorem{example}{Example}
\newtheorem{definition}{Definition}

\newcommand{\F}{\mathbb{F}}
\newcommand{\Fq}{\mathbb{F}_{q}}

\newcommand{\Fqn}{\mathbb{F}_{q^n}}
\newcommand{\Fqk}{\mathbb{F}_{q^k}}

\newcommand{\rank}{\mathrm{rank}}

\newcommand{\ds}{\mathrm{d_s}}
\newcommand{\calA}{\mathcal{A}}
\newcommand{\calB}{\mathcal{B}}
\newcommand{\calC}{\mathcal{C}}
\newcommand{\calD}{\mathcal{D}}
\newcommand{\calE}{\mathcal{E}}
\newcommand{\calF}{\mathcal{F}}
\newcommand{\calG}{\mathcal{G}}
\newcommand{\calL}{\mathcal{L}}

\newcommand{\calO}{\mathcal{O}}
\newcommand{\calP}{\mathcal{P}}
\newcommand{\calR}{\mathcal{R}}
\newcommand{\calS}{\mathcal{S}}

\newcommand{\calU}{\mathcal{U}}
\newcommand{\calW}{\mathcal{W}}
\newcommand{\calV}{\mathcal{V}}
\newcommand{\calX}{\mathcal{X}}
\newcommand{\calY}{\mathcal{Y}}

\newcommand{\fred}{f_{\mathrm{red}}}
\newcommand{\fexp}{f_{\mathrm{exp}}}
\newcommand{\fmul}{f_{\mathrm{mul}}}
\newcommand{\pr}{\mathrm{P}}

\newcommand{\gbinom}[2]{\genfrac{[}{]}{0pt}{}{#1}{#2}}

\DeclareMathOperator{\Inv}{Inv}

\DeclareMathOperator{\Gr}{Gr}

\counterwithout{table}{section}

\newcommand{\upmodels}{\mathrel{\rotatebox[origin=c]{90}{\scalebox{0.7}{$\models$}}}}

\newcommand{\Span}[2]{\langle #2 \rangle_{#1}}
\newcommand{\Core}[1]{\mathop{\scalebox{1.5}{$\Cap$}}\nolimits_{\mathbb{F}_{q^k}}(#1)}

\newcommand{\Wildcard}{*}
\newcommand{\Slice}[3]{
    \IfEqCase{#2}{
        {1}{#1_{#3, \Wildcard, \Wildcard}}
        {2}{#1_{\Wildcard, #3, \Wildcard}}
        {3}{#1_{\Wildcard, \Wildcard, #3}}
    }
}

\title{Decoding Desarguesian spread codes beyond \\half minimum distance}
\author[1]{Ermes Franch%
  \thanks{Co-funded by the European Union, Grant Agreement No.~101126560, through the Bergen Research and Training Programme for Future AI Leaders Across the Disciplines (LEAD AI).}}
\author[1]{Chunlei Li}
\author[2]{Angelica Piccirillo}

\affil[1]{University of Bergen, Bergen, Norway}
\affil[2]{Technical University of Munich, Munich, Germany}

\date{}

\begin{document}

\maketitle
\begin{abstract}
Spread codes are a well-known family of constant-dimension subspace-metric codes.
For constant dimension $k$ and ambient space dimension $n$ being a multiple of $k$, these codes have minimum distance $2k$ and a rich geometric structure.
In this paper, we study the decoding capabilities of the Nearest Neighbor Decoder for Desarguesian spread codes, establishing that unique decoding is still achievable beyond half the minimum distance. Motivated by this, we develop a new decoding algorithm to uniquely decode Desarguesian spread codes in the presence of both insertions and deletions, which increase and decrease, respectively, the dimension of the
transmitted codeword. Even when the sum of the dimensions of insertions and deletions exceeds half the minimum distance, provided that deletions are of dimension at most $k-2$, the algorithm succeeds with a small decoding failure. We also propose two refinements to this algorithm that, empirically, can handle nearly as many insertions as the Nearest Neighbor Decoder.
\end{abstract}
\textbf{Keywords: }Subspace codes, Desarguesian spread codes, Nearest neighbour decoding, Product of subspaces, Generalized evasive subspaces. 
\section{Introduction}
K\"{o}tter and Kschischang in~\cite{Kotter08} introduced the idea of using subspace codes for random linear network coding. In this setting, an operator channel takes in a vector space and
puts out another vector space, possibly with erasures 
(deletion of dimensions due to, e.g., an insufficient
min-cut in the network or an unfortunate choice of coefficients
in the random linear network code) and errors (insertion of dimensions due to errors or deliberate malfeasance).
Subspaces of $\Fq^n$ (equivalently $\Fq$-linear subspaces of $\Fqn$) are typically considered.
Using the subspace metric, they showed that a nearest neighbour (equivalently minimum distance) decoder can be used to reliably decode a corrupted codeword, namely, finding a subspace in the code that is the closest to the received subspace when their intersection is sufficiently large.
In the construction of subspace codes, K\"{o}tter and Kschischang proposed to use lifted Gabidulin codes~\cite{Gabidulin1985,silva2008rank} in the following way. 
A lifted Gabidulin code is derived by appending the identity matrix $I \in \Fq^{m \times m}$ to all the codewords in a Gabidulin code $\calG \subseteq \Fq^{m \times n}$ in its matrix form. 
The space $\calV=\{\mathrm{Rowspan}(I\mid C)\colon C \in \calG \}$ will be the lifting of the Gabidulin code $\calG$. The close connection to Gabidulin codes enables an efficient bounded-distance decoder for this class of subspace codes~\cite{silva2008rank}.
More precisely, the decoder can successfully recover the transmitted codeword from any received subspace whose distance from the original codeword is less than half the minimum distance of the code. 
Researchers also considered list decoding for rank-metric codes and subspace codes beyond half of the minimum distance (see~\cite{Bartz15} and reference therein), which typically gives lists of exponential size for any radius beyond half of the minimum rank distance~\cite{wachter13}.

Existing works on subspace codes have mainly focused on variants of Gabidulin codes or folded Gabidulin codes.
This is largely owing to the fact that Gabidulin-like codes provide  not only flexible code rates but also efficient decoding algorithms, like syndrome-based decoding and interpolation-based decoding~\cite{Bartz2022,Gabidulin2022}. 
Another interesting family of subspace codes is based on the construction of spreads in finite geometry, see for example~\cite{Hirschfeld1998}. 
For a constant dimension $k$, spread codes offer a larger minimum distance $2k$ yet a smaller rate compared to lifted Gabidulin codes, with minimum distance $k$. 
The use of a particular family of these codes known as \textit{Desarguesian spread codes} in the context of random linear network decoding was proposed for the first time by Gorla, Manganiello and Rosenthal in~\cite{Gorla08}.
They designed two decoding algorithms (see~\cite{Gorla11,Gorla08}) that can recover uniquely a $k$-dimensional codeword $\calC$ from a received space $\calR$ within
half of the minimum distance from $\calC$ in polynomial time; more precisely when the received space $\calR$ has dimension at most $k$ and $\dim(\calR \cap \calC) \geq \frac{k+1}{2}$.

 \subsection{Our contribution}
In this work, we propose a probabilistic polynomial-time decoding algorithm for Desarguesian spread codes. Our algorithm can decode subspaces beyond the unique-decoding radius with a success probability quickly converging to one as the field size increases.
In particular, our approach can handle the case where the received subspace has dimension larger than the codeword, i.e., when insertions outnumber deletions. This addresses the open problem raised in the conclusion of~\cite{Gorla11}. Moreover, our results show that it is possible to decode with high probability (w.h.p.) received spaces affected simultaneously by large deletions with dimension $d$ up to $k-2$, and by insertions with dimensions exceeding $d$, up to the theoretic decoding radius of a nearest neighbour decoder.

We also present two refinements of the basic algorithm that improve its performance, both in terms of the dimension of correctable insertions and experimental decoding success, as supported by the experimental results reported in Section~\ref{sec:experimental_results}.

In Section \ref{sec:Probability of success for a minimum distance decoder}, we further link the nearest neighbour decoder of Desarguesian spread codes to the geometric notion of subspaces that are evasive with respect to Desarguesian spreads. This provides an intrinsic limit on the dimension of random insertions that can be decoded reliably from the received subspace alone.
Experimental results show that the refined version of our algorithm starts to degrade only near the same regime where nearest neighbour decoding itself becomes unreliable.

 \subsection{Outline of the paper} 
 The structure of the paper is as follows. Section \ref{sec: preliminaries} recalls several tools that will be employed throughout the paper and provides an overview of subspace codes, with a particular focus on Desarguesian spread codes. In Section \ref{sec: ER dec}, we recall the two types of errors that may occur for subspace codes in random network coding, namely insertions and deletions, and introduce the expanding and reducing functions on $\Fq$-subspaces of $\mathbb{F}_{q^{rk}}$. We then present our decoding algorithm, which relies heavily on these operations and, in particular, on their behavior with respect to $\mathbb{F}_{q^k}$-linearity. Section \ref{sec: Improved techniques: Expand Reduce Expand and Filtered ERE} presents two refined versions of the  aforementioned algorithm, and Section \ref{sec:experimental_results} summarizes the results of the experiments conducted on the three algorithms. In Section \ref{sec:Probability of success for a minimum distance decoder}, we study Desarguesian spread decoding from the perspective of a nearest neighbour decoder, highlighting its connection with well-known geometric objects. In Section \ref{sec: Conclusions and Open Problems}, we conclude this work along with future research directions arising from the results of this paper. Finally, the Appendix contains additional material that may help the reader better understand the work.

\section{Preliminaries}
\label{sec: preliminaries}
Throughout this paper, let $q$ denote a prime power, $\Fq$ the finite field with $q$ elements, and $\Fqn$ an extension field of $\Fq$ of degree $n$. We write $\Fqn^*=\Fqn\setminus\{0\}$. For an $\Fq$-subspace $\calW$ of $\Fqn$ and $a\in\Fqn^*$, we denote by $a\calW$ the $\Fq$-subspace $\{aw : w\in\calW\}$. In particular, if $\calW=\Fqk$, we simply write $a\Fqk$.
\par Consider the set $\calL(\Fq^n)$ that consists of all subspaces of $\Fq^n.$
When fixing a basis of $\Fqn$ over its subfield $\Fq$, this set is equivalent to the set of all the $\Fq$-linear subspaces of $\Fqn$, denoted by $\calL_{\Fq}(\Fqn)$ . 
We can equip the lattice $\calL_{\Fq}(\Fqn)$ with the subspace distance~\cite{Kotter08}.
\begin{definition}[Subspace distance]
    The subspace distance between two $\Fq$-linear subspaces $\calC,\calD \subseteq \Fqn$ is defined as:
    \[
     \ds(\calC,\calD) =  \dim(\calC + \calD) - \dim(\calC \cap \calD) = \dim(\calC) + \dim(\calD) - 2 \dim(\calC \cap \calD).
    \]
\end{definition}

A subspace code is simply a subset of $\calL_{\Fq}(\Fqn).$
Let us denote by $\Gr_{\Fq}(k,\Fq^n)$ the set of all $k$-dimensional $\Fq$-linear spaces of $\Fq^n.$
This is isomorphic to the set $\Gr_{\Fq}(k, \Fqn)$
of all the $\Fq$-linear spaces in $\Fqn$ of $\Fq$-dimension $k$. 
A significant class of subspace codes is that of constant-dimension codes~\cite{Kotter08},
consisting of subspaces in $\Gr_{\Fq}(k, \Fqn)$, equivalently in $\Gr_{\Fq}(k,\Fq^n)$, and an important subclass of constant-dimension codes is given by the so-called spread codes~\cite{Gorla08}.
\begin{definition}[Spread Code] \label{def:spread_code}
   A spread code is a constant-dimension code $\bm \calC\subseteq \Gr_{\Fq}(k, \Fqn)$ such that any distinct $\calC_1, \calC_2 \in \bm \calC$ have a trivial intersection $\calC_1 \cap \calC_2 = \{ 0\}$
   and, for each $a \in \Fqn^*,$ there exist $\calC_i$ such that $a \in \calC_i.$
\end{definition}
Spread codes have two attractive properties:
since each intersection between two elements of a spread code is trivial, every spread code with constant dimension $k$ has  minimum distance $2k$.
This is the largest distance achievable for a constant-dimension subspace code of dimension $k$.
In addition, as spread codes cover $\Fqn^*$, they also achieve the highest possible cardinality, given by $(q^n -1)/(q^k -1)$, among constant-dimension codes of minimum distance $2k$.
This number is easily obtained by the fact that $|\Fqn^*| = q^n - 1$  while each subspace contains exactly $q^k-1$ non-zero vectors.
As all the intersections between the subspaces of a spread are trivial, the cardinality of a spread code is $(q^n -1)/(q^k -1)$, this also indicates that we can construct a spread only if $k \mid n.$
It is to be noted that while the ambient space $\Fqn$ is covered by a spread code, this fact does not indicate perfect covering property; namely, it is not true that all the subspaces of $\Fqn$ are contained in a ball of half-minimum subspace distance $k$ from the codewords of a spread~\cite{Martin1995}.

In this paper we will focus on the subclass of \textit{subfield spread codes} or \textit{Desarguesian spread codes}, in view of their algebraic structure and geometric interpretation.
\begin{definition}[Desarguesian spread code]
\label{def:subfield_spread}
   Let $k|n$ and consider the subfield $\Fqk \subseteq \Fqn$ of $q^k$ elements.
   We can consider the Desarguesian spread code given by the orbit of the subfield $\Fqk$ under the action of the multiplicative group $\Fqn^*$ as follows,
   $$
    \bm  \calD = \{ a \Fqk \mid a \in \Fqn^*\} \subseteq \Gr_{\Fq}(k,\Fqn).
   $$
\end{definition}
It is easy to see that the spaces $a \Fqk, b \Fqk$ are either the same space or disjoint.
Indeed, if $c \in a \Fqk \cap b \Fqk,$ then $c = a x_1 = b x_2$ for some $x_1,x_2 \in \Fqk,$ implying $a = b x_2 x_1^{-1} \in b \Fqk$ and $b = a x_1 x_2^{-1} \in a \Fqk$.
As mentioned in the introduction, the use of these codes in random linear networks  was proposed in~\cite[Definition 2]{Gorla08},\cite[Lemma 5, Theorem 6]{Gorla11} together with an efficient bounded-distance decoding algorithm, which returns a unique codeword whenever the received space lies within the distance $k-1,$ strictly below the half minimum of $\bm \calD$, which equals $k$.
\\
\par In this work, we consider the nearest neighbour decoder (also known as minimum distance decoder) in the specific case of Desarguesian spread codes.
\begin{definition}[Nearest Neighbour Decoder]
    \label{def: Nearest neighbour decoder}
    Let $n=kr$ and let $
    \bm  \calD$ be a Desarguesian spread code of constant dimension $k$.
    For a received subspace $\calR \subseteq \Fqn$ a
    \textbf{Nearest Neighbour Decoder (NND)} will find the codeword $\hat{\calC}$ as
    \[
   \hat{\calC} = \arg \min_{\calC \in \bm \calD} \{ \ds(\calR,\calC)\}. \]
    In other words, find $\calC \in \bm \calD$ that minimize the distance $\ds(\calR,\calC).$
    If the distance between the original codeword $\calC_0$ and the received space $\calR$ is bounded by
    $\ds(\calR,\calC_0) <\frac{\ds(\bm  \calD)}{2}=k$, a nearest neighbour decoder is guaranteed to return $\calC_0$.
\end{definition}
With this definition in mind, we recall that the algorithm presented in~\cite{Gorla11} for the Desarguesian spread code $\bm\calD$ of constant dimension $k$ uniquely decodes up to distance $k-1$, i.e., up to half the minimum distance, requiring $O((n-k)k^3)$ operations over the subfield $\Fqk$. However, it is restricted to the case where the received subspace $\calR$ satisfies $\dim(\calR)\leq k$, due to the theoretical foundation upon which it is based.
\\
\par Note that if the received subspace is at distance greater than $k-1$ from the transmitted codeword, it may be closer or equally close to another codeword than to the one that was transmitted.
For example, let $k=2s$ and suppose that we transmit the space $a \Fqk$, if the received subspace is subject to $s$ deletions (meaning $s$ linearly independent vectors are lost during the transmission) and $s$ insertions (meaning $s$ linearly independent vectors are added) its total distance from $a \Fqk$ will be $$\ds(\calR,a \Fqk)=\dim(\calR) + \dim(a \Fqk) - 2 \dim(\calR \cap a \Fqk)=k+k-2s=k.$$
If we further assume that all the insertions come from the same space $b \Fqk \neq a \Fqk$, the received space will be at the same distance $k$ from both $a \Fqk$ and $b\Fqk.$ In this case, the output of an NND will not be unique.
Notice that, to construct such a negative example, we had to choose the insertions in a very specific way. A typical insertion is unlikely to be contained in a subspace $b\Fqk$; therefore, the space $a\Fqk$ will still be the closest choice, with high probability, even when the distance is $k$. This will be discussed in more detail in Section \ref{sec:Probability of success for a minimum distance decoder}, where we study the relationship between the geometric properties of the received space and the success probability of the NND in Definition \ref{def: Nearest neighbour decoder}. This suggests that a NND can still recover the correct codeword even for distance larger than $k-1$ from the original codeword with some probability.

Inspired by this possibility, we propose a probabilistic decoding algorithm that can decode beyond half the minimum distance and also handle the case where the received subspace has dimension larger than $k$, thereby addressing the open question in~\cite{Gorla11}.
\section{Expansion-Reduction Decoding of Desarguesian Spread Codes}
\label{sec: ER dec}
\subsection{Error Model for Subspace Codes}
\label{subsec: error model}
We start by recalling the two types of errors, namely, insertion and deletion (corresponding to errors and erasures in~\cite[Sec.~III-C]{Kotter08}) that may occur in the context of subspace codes. Let $\calC = \Span{\Fq}{c_1,c_2,\dots,c_k}$\footnote{Unless otherwise stated, the notation $\Span{\Fq}{v_1,v_2,\dots,v_{\ell}}$ denotes the $\Fq$-subspace of $\Fqn$ spanned by the elements $v_1,v_2,\dots,v_{\ell}\in\Fqn$.} be a codeword in a constant-dimension subspace code $\bm \calC$.
\begin{itemize}
	\item \textit{Insertion.} 
	An insertion of weight $t$ corresponds to receiving $\calC + \Span{\Fq}{b_1, \ldots, b_t}$ where $\Span{\Fq}{b_1, \ldots, b_t}$ is a subspace of dimension $t.$
	\item \textit{Deletion.} A deletion of weight $d$ corresponds to receiving a subspace $\calU = \Span{\Fq}{c_1,\ldots,c_{k-d}} \subseteq \calC$ of dimension $k-d$ instead of the whole space $\calC$ of dimension $k.$  
	
\end{itemize}
We refer to $t$ insertions (resp. $d$ deletions) when an insertion of weight $t$ (resp. a deletion of weight $d$) has occurred.

\par When only one of these two types of errors occurs, it is relatively easy to recover the original codeword $\calC$ by applying just the reducing function (to remove insertions) or the expanding function (to recover from deletions), which we will define later.
Decoding is more challenging when both types of error occur at the same time.
In this case, we need to consider the received space of the form
$$
\calR = \calU + \calB = \Span{\Fq}{c_1, \ldots, c_{k-d}} + \Span{\Fq}{b_1, \ldots, b_{t}},
$$
where $\calU$ is a $(k-d)$-dimensional subspace of the intended codeword $\calC$ and $\calB$ is a random $\Fq$-linear subspace of $\Fqn$ of dimension $t.$ The distance between $\calC$ and $\calR$ is given by 
$$\ds(\calC, \calR)=\dim(\calC) + \dim(\calR) - 2 \dim(\calC \cap \calR)\leq d+t.$$
Throughout the theoretical analysis on which the algorithm is based, without loss of generality we assume $\calR = \calU \oplus \calB$ and hence $\ds(\calC, \calR)=d+t$. 
In fact, let $r= \dim(\calR)$, we can always consider a basis of the space $
\Tilde{\calU} = \calC \cap \calR$ given by $c_1, \ldots , c_h$ for some $h \geq k-d$.
This basis can be completed to a basis of $\calR$ with some elements $b_{h+1}, \ldots , b_{r}$ hence, if we denote by $\Tilde{\calB}$ the span of these elements, we have 
\[
\calR = \calU + \calB = \Tilde{\calU}\oplus\Tilde{\calB},
\]
where $\Tilde{\calU}=\Span{\Fq}{c_1, \ldots, c_{h}}$ and $\Tilde{\calB}=\Span{\Fq}{b_{h+1}, \ldots, b_{r}}$.
\subsection{Expansion and Reduction}
\label{sec:exp_and_red}
In this section, inspired by the functions presented in \cite{Aragon19} in the context of LRPC codes, we introduce two operations on $\calL_{\Fq}(\Fqn)$, called \textit{expansion} and \textit{reduction}, that preserve $\Fqk$-linear subspaces while significantly altering random $\Fq$-linear subspaces $\calB \subseteq \Fqn$.
We will also discuss the properties of the expansion and reduction operations, which pave the way for the decoding algorithms proposed in this paper.
\\
\par Recall that an element $\calC \in \bm \calD$ is of the form $\calC = a \Fqk=\Span{\Fq}{au_1, \ldots, au_k},$ for some $a \in \Fqn^*$ and $\Fq$-linearly independent $u_i\in\Fqk^*$. 
An important operation between the set of $\Fq$-linear subspaces of $\Fqn$ and one element of $\Fqk^*$ is the scalar multiple function 
\begin{align*}
    \fmul\colon \calL_{\Fq}(\Fqn) \times \Fqk^* &\rightarrow \calL_{\Fq}(\Fqn)\\
    (\calV,x)&\mapsto\fmul(\calV,x) \coloneqq x \calV.
\end{align*}
Thanks to the distributive property of the product, this operation preserves the vector space structure as well as the $\Fq$-dimension of the space.
\begin{definition}[{Expanding function}]
The expanding function $\fexp:\,$
$
\calL_{\Fq}(\Fqn) \times (\Fqk^*)^s 
\rightarrow \calL_{\Fq}(\Fqn) $ is defined
as $$\fexp(\calV, \bm x)=x_1 \calV + \cdots + x_s \calV,$$
where $\bm x=(x_1,\ldots,x_s)$. Moreover, $s$ will be called the length of the expansion, while $\fexp(\calV, \bm x)$ will be called $s$-expansion, or simply expansion when the length is clear from the context.
\end{definition}
Equivalently, the expansion function $\fexp(\calV, \bm x)$ can be defined as the product subspace between $\calV$ and the $\Fq$-linear space generated by the entries of $\bm x$.
Let $\calX = \Span{\Fq}{x_1,\ldots,x_s} \subseteq \Fqk$ and let $\calV = \Span{\Fq}{v_1, \dots, v_{h}}$, then
$$
    \fexp(\calV,\bm x) = \calV.\calX \coloneqq \Span{\Fq}{\{v_i x_j \mid i \in [h], j \in [s]\}},
$$
which is the smallest $\Fq$-linear subspace containing the set $\calV \calX := \{ vx \mid v \in \calV, x \in \calX\}.$ This product was introduced in the context of LRPC codes in \cite{gaborit2013} and its properties were further analyzed in \cite{Aragon19}. 
\\
\par Consider a subspace $\calU \subseteq \calC = a \Fqk$ and let $x \in \Fqk^*$. Then $x \calU$ has the same $\Fq$-dimension as $\calU$, and since $x\calC=\calC$, it follows that $x \calU \subseteq \calC$.
This implies that the expanding function we defined above preserves $\calC= a \Fqk$ while, for any subspace $\calU$ of $\calC$, we have $\fexp(\calU, x)\subseteq \calC.$ 
In principle, the $\Fq$-dimension of the expanded subspace could be greater than $\calU$ while it will be upper bounded by $\min\{s \dim(\calU), k \}.$ In our experiments, we observed that $\dim(\fexp(\calU, \bm x))=\min\{s \dim(\calU), k \}$ holds in most cases, in agreement with the theoretical results of \cite{Aragon19}. When $\dim(\fexp(\calU, \bm x))=\min\{s \dim(\calU), k \}$, we say that $\fexp(\calU, \bm x)$ is an \textit{optimal expansion}.
\par In \cite{Aragon19}, given two random $\Fq$-subspaces $\calA$ and $\calB$ of $\Fqn$ with dimensions $\alpha$ and $\beta$, respectively, the authors, under the assumption that $\alpha\beta < n$, investigate the typical dimension of the product subspace $\calA.\calB$. 
\begin{proposition}[{\cite[Proposition III.3]{Aragon19}}]\label{prop: expected dimension of the expansion}
    Let $\calB$ be a fixed $\Fq$-subspace of $\Fqn$ of dimension $\beta$ and let $\calA=\langle a_1,\ldots,a_{\alpha}\rangle_{\Fq}$ where $a_1,\ldots,a_{\alpha}$ are $\alpha$ $\Fq$-linearly independent elements of $\Fqn$ chosen uniformly at random such that $\alpha\beta < n$.  Then $\dim(\calA.\calB)=\alpha\beta$ with probability at least $1-\alpha q^{-(n-\alpha\beta)}$.
\end{proposition}
This yields the following corollary, whose proof is a straightforward reformulation of the previous proposition. For completeness, it is provided in Appendix \ref{sec: proof of corollary 1}.
\begin{corollary}
    \label{cor: expected dimension of the expansion}
    Let $\calU = \Span{\Fq}{u_1, \dots, u_{k-d}}\subseteq a\Fqk$ with $a\in\Fqn^*$ be a fixed subspace and suppose we construct a random subspace $\calX = \Span{\Fq}{x_1,\ldots,x_s} \subseteq \Fqk$ by choosing uniformly at random $s \leq k$ $\Fq$-linearly independent elements $x_1,\ldots,x_s$ of $\Fqk$. Let $s \dim(\calU)<k$, then $\dim(\fexp(\calU, \bm x))=\dim(\calU.\calX)=s \dim(\calU)$ with probability at least $1-s q^{-(k-s(k-d))}$.
\end{corollary}
In the worst case of $k-2$ deletions, i.e., when $\dim(\calU)=2$, the probability that $\dim(\fexp(\calU, \bm x))=2s$ can be refined to $1 - q^{2s - k - 1}$.
The proof is rather technical and can be found in Appendix \ref{app: Expansion}. More generally, the analysis in Appendix~\ref{app: Expansion} applies to specific two expansions of an arbitrary $\Fq$-subspace $\calV$ of $\Fqk$, i.e., to spaces of the form $\calV + a \calV$ with $a\in\Fqk^*$. Therefore, in the case in which $\dim(\calU)=2$, thanks to the commutativity of the product $\calU.\calX = \calX.\calU$, the same analysis describes the behavior of $\fexp(\calU,\bm x)$ when $\dim(\calU)=2$.
In this regime, non-optimal expansions correspond to elements $u \in \Fqk$ such that $\calX \cap u\calX \neq \{0\}$, i.e., $u$ lies in intersections of subspaces of the form $x^{-1}\calX$ for $x \in \calX$.
When these intersections are minimal (typically equal to $\Fq$), non-optimal expansions are relatively frequent but only miss optimality by one dimension.
Conversely, if larger intersections occur, leading to a loss of multiple dimensions, then the set of such $u$ is necessarily much smaller.
This trade-off is favorable for our decoding algorithm: severe deviations from optimal expansion are rare, while the more common non-optimal cases only incur a limited loss in dimension.
As a consequence, when $\dim(\calU)=2$, the expansion is optimal with high probability.
\\
\par We now return to the general setting. As an immediate consequence of Corollary~\ref{cor: expected dimension of the expansion}, in the presence of only $d$ deletions, if $\dim(\calU)=k-d \geq k/s$, it is likely that $\fexp(\calU, \bm x) = \calC$ for a vector $\bm x \in (\Fqk^*)^s$ chosen uniformly at random.

\par When we consider both deletion and insertion errors, instead, it is important to understand the behavior of the expanding function on the subspace 
$
    \calR = \calU \oplus \calB,
$
where $\calU \subseteq a\Fqk$ has dimension $k-d$ and $\calB \subseteq \Fqn$ is a random subspace of dimension $t.$
In this case we have
$$
    \fexp(\calR, \bm x) = \calR .\calX = \calU.\calX + \calB.\calX
$$
and, if $\dim(\calU) = k-d \geq k/s$, then it is likely that
$$
   \fexp(\calR,\bm x) =\calC + \calB',
$$
where $\calB'=\calB.\calX $ and $\dim(\calB') \leq \min \{st, n \}.$ 
More generally, with an expansion of length $s$, the dimension of the expansion $\fexp(\calR, \bm x)$ is upper bounded by
\begin{equation}
    \label{eq: dim_exp}
    \min\{ st + k, st + s(k-d), n\},
\end{equation}
where we notice that
\begin{equation}
    \label{eq: min upper bound on dimension of expansion}
    \min\{st + k,\, st + s(k-d)\}
=
\begin{cases}
st + k & \text{if } s \ge \frac{k}{k-d}, \\
st + s(k-d) & \text{otherwise}.
\end{cases}
\end{equation}
The above observations show that the expanding function decreases the dimension of deletions, potentially recovering the entire codeword, while increasing the dimension of insertions. To this end, we introduce the following function.
\begin{definition}[Reducing function]
    The reducing function 
\(
\fred:\calL_{\Fq}(\Fqn)\times (\Fqk^*)^s\longrightarrow \calL_{\Fq}(\Fqn)
\)
is defined as
\[
\fred(\calV,\bm y)
   = \bigcap_{\substack{i\in [s]}} y_i\calV,
\]
where $\bm y=(y_1,\ldots y_s)$. Moreover, $s$ will be called the length of the reduction, while $\fred(\calV, \bm x)$ will be called $s$-reduction, or simply reduction when the length is clear from the context.
\end{definition}
As $y_i \Fqk= \Fqk$ for any $y_i \in \Fqk^*$ then $\fred(a \Fqk, \bm y) = a \Fqk,$ which means that the reducing function keeps each codeword  $\calC$ unchanged.
On the other hand, the dimension of a random space $\calV$ can decrease, possibly even to zero. The following proposition shows that this is more likely when $\dim(\calV)\ll \left\lfloor\frac{n}{2}\right\rfloor$.
\begin{remark}
    \label{prop: probability reduction goes to zero}
    Let $\calV$ be a random $\Fq$-subspace of $\Fqn$. Let $y_1,\ldots,y_s$ be $\Fq$-linearly independent elements chosen uniformly at random from $\Fqk$. Then 
    \[
    \mathrm{P}\left(\bigcap_{i\in [s]} y_i\calV=\{0\}\right)\geq 1-q^{2\dim(\calV)-n},
    \]where $[s]:=\{1,2,\dots,s\}$. Indeed note that if there exists $i,j\in[s]$ such that $i\neq j$ and $y_i\calV\cap y_j\calV=\{0\}$, then $\bigcap_{i\in [s]} y_i\calV=\{0\}$. Consequently 
    \[
    \mathrm{P}(\text{There exist }i,j\in[s]\text{ such that }i\neq j\text{ and }y_i\calV\cap y_j\calV=\{0\})\leq \mathrm{P}\left(\bigcap_{i\in [s]} y_i\calV=\{0\}\right).
    \]
    For $i\neq j$, we have that $y_i\calV\cap y_j\calV\neq\{0\}$ if and only if $\calV\cap y_i^{-1}y_j\calV\neq\{0\}$ if and only if $\calV\cap y\calV\neq\{0\}$ with $y=y_i^{-1}y_j$. Since $y_1,\ldots,y_s$ are $\Fq$-linearly independent elements chosen uniformly at random in $\Fqk$, the elements $y_{ij}=y_i^{-1}y_j$ are also distributed uniformly at random.
    Therefore, by considering the complement of the events,
    \begin{align*}
        \mathrm{P}\left(\bigcap_{i\in [s]} y_i\calV\neq\{0\}\right)&\leq\mathrm{P}(\text{For all }i,j\in[s]\text{ such that }i\neq j\text{, we have }y_i\calV\cap y_j\calV\neq\{0\})\\
        &=\mathrm{P}(\text{For all }i,j\in[s]\text{ such that }i\neq j\text{, we have }\calV\cap y_{i,j}\calV\neq\{0\})\\
        &\leq \min_{i,j\in[s],i\neq j}\mathrm{P}(\calV\cap y_{i,j}\calV\neq\{0\})\leq q^{2\dim(\calV)-n}.
    \end{align*}
    The last inequality follows from the fact that since $\calV$ is a random subspace of $\Fqn$, i.e., its elements are chosen uniformly at random from $\Fqn$, once we have fixed $v\in\calV$, the probability that for an element $y_{i,j}\in\Fqk^*$ chosen uniformly at random we also have $y_{i,j}v\in\calV$, is 
    \[
    \mathrm{P}(y_{i,j}v\in\calV)=\frac{\lvert\calV\rvert}{q^n},
    \]
    hence
    \begin{align*}
        \mathrm{P}(\calV\cap y_{i,j}\calV\neq\{0\})&=\mathrm{P}(\text{There exists }v\in\calV\text{ such that }y_{i,j}v\in\calV)\\
        &=\mathrm{P}\left(\bigcup_{v\in\calV}y_{i,j}v\in\calV\right)\\
        &\leq \frac{\lvert\calV\rvert^2}{q^n}=q^{2\dim(\calV)-n}.
    \end{align*}
\end{remark}
Observe that the lower bound in Remark~\ref{prop: probability reduction goes to zero} does not depend on $s$. In fact, the proof only relies on the case $s=2$. Since adding more subspaces to the intersection can only decrease its dimension, the probability that
\[
\bigcap_{i\in[s]} y_i\calV=\{0\}
\]
actually increases with $s$.
\\
\par
In reductions of length two observe that, if $\bm y = (y_1,y_2)$ we need $y_1,y_2$ to be $\Fq$-linearly independent as otherwise $y_1 \calV = y_2 \calV.$
It would be tempting to generalize this relation assuming the reductions depend only on the subspace generated by the entries of $\bm y$.
The following example shows this is not the case.
Consider the space $\calV =\Span{\F_2}{1,\alpha} \subseteq \F_{8}$ such that $\alpha^3+\alpha+1=0$ and the vector $\bm y = (1,\alpha,1+\alpha)$ having support $\calV$ of dimension $2$.
The reduction  $\fred(\calV,\bm y) = \calV \cap \alpha \calV \cap (\alpha+1) \calV  = \{ 0\}.$ 
If we consider the vectors $(1,\alpha)$ and $(1, \alpha+ 1)$, those have the same support as $\bm y$ but reducing $\calV$ by these shorter vectors leads to $\{0, \alpha \}$ and $\{0, \alpha + 1 \}$, respectively. 
\subsubsection{$\Fqk$-linear subspaces with respect to expansion and reduction}
\label{subsec: From expansion and reduction to Fqk-linearity}
For any given $\Fq$-linear space $\calV \subseteq \Fqn$ where $n=kr$,
we can always consider the smallest $\Fqk$-linear space that contains $\calV$ and the largest $\Fqk$-linear space contained in $\calV$.
\begin{definition}
    Let $n=kr$ and $\calV$ be an $\Fq$-linear space of $\Fqn$. We denote by $\Span{\Fqk}{\calV}$ the $\Fqk$-span of $\calV$, i.e., the smallest $\Fqk$-linear space containing $\calV$.
    We denote by $\Core{\calV}$ the largest $\Fqk$-linear space contained in $\calV$, namely $\displaystyle\bigcap_{b\in\Fqk^*}b\calV.$ 
\end{definition}
The existence follows from the observation that $ \{0\} \subseteq \calV \subseteq \Fqn$. The uniqueness of $\Core{\calV}$ comes from the following simple argument: let $\calW_1,\calW_2$ be two $\Fqk$-linear spaces such that $\calW_1, \calW_2 \subseteq \calV$, then $\calW_1,\calW_2\subseteq\calW_1 + \calW_2\subseteq\calV$. The uniqueness of $\Span{\Fqk}{\calV}$ comes from a similar argument: let $\calW_1,\calW_2$ be two $\Fqk$-linear spaces such that $\calV \subseteq \calW_1, \calV \subseteq\calW_2$, then $\calV\subseteq \calW_1 \cap \calW_2\subseteq\calW_1,\calW_2$.

There is a close connection between these two spaces and the functions presented in Section \ref{sec:exp_and_red}.
In particular, whenever $\bm b$ is a basis of $\Fqk$,
the expansion $\fexp(\calV, \bm b)$ transforms $\calV$ in $\Span{\Fqk}{\calV}$. This represents a limit to the expansion as $\Fqk$-linear spaces are not affected by further expansions. 
\begin{lemma}
\label{lm:exp_to_Fqk_linear}
    Let $\bm b = (b_1,\ldots, b_k)$ be an ordered $\Fq$-basis of  $\Fqk$. Then, for an $\Fq$-subspace $\calV\subseteq\Fqn$ we have 
    \[
    \calV\subseteq \fexp(\calV, \bm b) = b_1 \calV + \cdots + b_k \calV = \Span{\Fqk}{\calV}.
    \]
\end{lemma}
\begin{proof}
    We need to show that $\fexp(\calV, \bm b)$ is the smallest $\Fqk$-linear subspace containing $\calV$, namely, that $\calV \subseteq \fexp(\calV, \bm b)$ and $\fexp(\calV, \bm b)$ is $\Fqk$-linear. Let $v \in \calV$, as $1 \in \Fqk$ we can write it as $\sum_{i=1}^k c_i b_i$ for some $c_i \in \Fq$, then $v = \sum_{i=1}^k c_i b_i v$ where all the addends in the sum belong to some space of the form $b_i \calV$, which proves $\calV \subseteq \fexp(\calV, \bm b).$
    
    To show the $\Fqk$-linearity we only need to prove that $\lambda \fexp(\calV, \bm b) \subseteq \fexp(\calV, \bm b)$ for any $\lambda \in \Fqk$ as $\fexp(\calV, \bm b)$ is trivially closed under addition.
    Let $w = \sum_{i=1}^{k} v_i b_i \in \fexp(\calV, \bm b)$ for some $v_i \in \calV$ and let $\lambda \in \Fqk$ .
    Then $\lambda w = \sum_{i=1}^k  v_i \lambda b_i$. As $\lambda b_i \in \Fqk$, it can be written as $\lambda b_i = \sum_{j=1}^k \lambda_{i,j} b_j$ for some appropriate coefficients $\lambda_{i,j} \in \Fq$, hence $\lambda w = \sum_{i=1}^k \sum_{j=1}^k  \lambda_{i,j} v_i  b_j =  \sum_{j=1}^k  \hat{v_j}  b_j$ where $\hat{v_j} = \sum_{i=1}^k \lambda_{i,j} v_i \in \calV$.
\end{proof}
Note that, in general, it is not true that $\calV \subseteq \fexp(\calV, \bm b)$ if $\bm b$ is not an $\Fq$-basis of $\Fqk$.
\\
\par
There is a similar connection between the function $\fred(\calV,\bm b)$ and the space $\Core{\calV}$.

\begin{example}
    Let $\F_{16}=\F_{2}(\alpha)$ where $\alpha^4=\alpha+1$, i.e., $\mathbb{F}_{16}=\langle 1,\alpha,\alpha^2,\alpha^3\rangle_{\F_2}$ and $\F_4=\langle 1, \alpha+\alpha^2\rangle_{\F_2}$. Let us consider $\calV=\langle \alpha,\alpha^2,\alpha^3\rangle_{\F_2}$, then 
    \[
    \calV\cap(\alpha+\alpha^2) \calV= \langle \alpha\rangle_{\F_4}
    \]
    is the largest $\F_4$-linear subspace contained in $\calV$.
\end{example}
It is easy to see that applying a reduction to a space that is already $\Fqk$-linear has no effect.
This implies that if $\calV$ is an $\Fq$-linear space containing an $\Fqk$-linear space $\calW$, any reduction of $\calV$ will still contain $\calW$.
The next Lemma gives a sufficient condition on $\bm b$ for which $\fred(\bm b, \calV) = \Core{\calV}$.

\begin{lemma}
\label{lm:random_red_to_Fqk_linear}
    Let $\bm b= (b_1, \ldots, b_k)$ be an ordered $\Fq$-basis of $\Fqk$ with the property that $\bm b^{-1}= (b_1^{-1}, \ldots, b_k^{-1})$ is also an $\Fq$-basis of $\Fqk$. Then, for an $\Fq$-subspace $\calV\subseteq\Fqn$ we have
    $$
        \calW = \fred(\calV, \bm b)= \bigcap_{i=1}^k b_i \calV = \Core{\calV}.
    $$
\end{lemma}
\begin{proof}
    We start considering the special case where $\bm a = (1, \alpha, \ldots, \alpha^{k-1})$ is a polynomial basis.
    It can be proved that the vector $\bm a^{-1} = (1, \alpha^{-1}, \ldots, \alpha^{-k + 1})$ is still an ordered basis, in fact $\Span{\Fq}{\bm a^{-1}} = \alpha^{-k}\Span{\Fq}{\bm a}$. From the definition of $\calW$ we have
    $$
        \calW = \calV \cap \alpha \calV \cap \ldots \cap \alpha^{k-1} \calV.
    $$
    As $\calW$ is already $\Fq$-linear, to show it is also $\Fqk$-linear it is enough to prove that $\alpha \calW \subseteq \calW$.
    Let $w \in \calW$, then we have 
    \begin{equation}
    \label{eq:w=v0}
    w = v_0 = \alpha v_1 = \cdots = \alpha^{k-1} v_{k-1}
    \end{equation}
    for some $v_i \in \calV$, hence $\alpha w = \alpha v_0 = \cdots = \alpha^{k-1} v_{k-2} = \alpha^k v_{k-1}$
    from which we have $\alpha w \in \alpha \calV \cap \ldots \cap \alpha^{k-1} \calV$. 
    If we show that $\alpha w \in \calV$ we can conclude that $\alpha w \in \calW$.
    We know that $\alpha^k = \sum_{i=0}^{k-1} c_i \alpha^i$ for some $c_i \in \Fq$, thus we can substitute $\alpha^k$ in $\alpha w = \alpha^k v_{k-1}$, obtaining $ \alpha w =\sum_{i=0}^{k-1} c_i \alpha^i v_{k-1}$. From (\ref{eq:w=v0}), as $\alpha^{k-1} v_{k-1} = \alpha^{k-1-i} v_{k-1-i}$, we also have $\alpha^{i} v_{k-1} = v_{k-1-i} \in \calV$ it follows that $\alpha w = \sum_{i=0}^{k-1} c_i v_{k-1-i} \in \calV$.
    
    So far we proved that $\calW = \fred(\calV, \bm a)$ is an $\Fqk$-linear subspace of $\calV$. We are left with proving that it is the largest subspace of $\calV$ with this property.
    Let $\calS$ be an $\Fqk$-linear space such that $\calW \subseteq \calS \subseteq \calV$. Since $\alpha^i \calS = \calS$ for any $i$, we also have $\calS \subseteq \alpha^i \calV$ from which follows $\calS \subseteq \calW$ and this concludes the proof for the polynomial basis.
\\
    We now prove the general case by reducing it to the polynomial basis case just proved. In this setting, the chain of equations in (\ref{eq:w=v0}) becomes
    $$
        w = b_1 v_1 = \ldots = b_k v_k,
    $$
    for some $v_1,\ldots, v_k \in \calV$.
    Let us consider $\lambda\in\Fqk^*$. By assumption $(b_1^{-1}, \ldots, b_k^{-1})$ is a basis of $\Fqk$, hence $\lambda = \sum_{i=1}^{k} c_i b_i^{-1}$ for some $c_i \in \Fq$ and $\lambda w = \sum_{i=1}^k c_i b_i^{-1} w$. Substituting the suitable representation of $w$ in each term of the previous sum, we can write $$\lambda w = \sum_{i=1}^k c_i b_i^{-1} w = \sum_{i=1}^k c_i b_i^{-1} b_iv_i= \sum_{i=1}^k c_i v_i \in \calV.$$
    This implies $\lambda \calW \subseteq \calV$ for any $\lambda \in \Fqk^*$ which is equivalent to $\calW \subseteq \lambda^{-1} \calV$. 
    Choosing $\lambda \in\{ 1,\alpha^{-1},\ldots,\alpha^{-k+1}\}$ we obtain $\calW \subseteq \alpha^i \calV$ reducing the problem to the previous case.
\end{proof}

If we use a generic basis that do not respect the condition in Lemma \ref{lm:random_red_to_Fqk_linear} we are not guaranteed anymore to extract $\Core{\calV}$ from $\calV$.
We show this in the following example.
\begin{example}
    Let $\F_{64}=\F_{2}(\alpha)$ where $\alpha^6=\alpha+1$, i.e., $\mathbb{F}_{64}=\langle 1,\alpha,\alpha^2,\alpha^3,\alpha^4,\alpha^5\rangle_{\F_2}$ and $\F_8=\langle 1, \alpha+\alpha^2+\alpha^3,\alpha^3+\alpha^4\rangle_{\F_2}$. Let us consider $\calV=\langle 1,\alpha^3,\alpha^4,\alpha^5\rangle_{\F_2}$, then 
    \[
    \calV\cap(\alpha+\alpha^2+\alpha^3) \calV\cap(\alpha^3+\alpha^4) \calV= \langle 1+\alpha^3,1+\alpha^4\rangle_{\F_2}
    \]
    is an $\F_2$ subspace but not an $\F_8$ subspace.
\end{example}
Notice that, for any choice of $\bm b \in (\Fqk^*)^s$ we always have $\Core{\calV}\subseteq \fred(\calV, \bm b)$ so the reduction cannot go lower than that space. Moreover $\fred(\calV, \bm b) \subseteq \fred(\calV, \bm b')$ for any $\bm b'$ obtained from some puncturing of $\bm b$.
In practice, if we choose a vector $\bm b \in \Fqk^{k + \varepsilon},$ it is sufficient that there exists a sub vector $\bm b'\in \Fqk^{k}$ that is an $\Fq$-basis of $\Fqk$ satisfying the property of Lemma \ref{lm:random_red_to_Fqk_linear} to guarantee the extraction of $\Core{\calV}$.
Such bases are frequently observed in experiments. It is also important to stress that this condition is just sufficient but not necessary. Reducing through a long enough random vector will still lead to $\Core{\calV}$ in most of the cases.

\subsection{Expand and Reduce Decoding Algorithm}
\label{sec:A first decoding algorithm}
Recall that the decoding task for Desarguesian spread codes is as follows.
A codeword is of the form $a \Fqk$ for some $a \in \Fqn^*$ and 
a received space is of the form $\calR = \calU \oplus \calB$, where $\calU \subseteq a \Fqk$ is a subspace of $\Fq$-dimension $k-d$ and $\calB$ is a random $\Fq$-linear subspace of dimension $t$ with trivial intersection with $a \Fqk$. The parameters $d,t$ correspond to deletions and insertions and the subspace distance between the received space and the transmitted codeword is given by $\ds(\calR,a \Fqk) = d+t$.

\smallskip

The function $\fexp$ can be used to recover the original codeword in the case of \textit{pure} deletion. In particular, using Lemma \ref{lm:exp_to_Fqk_linear}, we have that $\fexp(\calU, \bm b) = a \Fqk$ whenever $\bm b$ is a basis of $\Fqk.$ Similarly, in the presence of \textit{only} insertions, from Lemma \ref{lm:random_red_to_Fqk_linear}, we could use $\fred(\calR,\bm b)= \fred(a \Fqk \oplus \calB, \bm b)$ with an appropriate vector $\bm b\in\Fqk^k$ to obtain the smallest $\Fqk$-linear space contained in $a \Fqk \oplus \calB$. Since in this case the codeword $a\Fqk$ is an $\Fqk$-linear subspace contained in $\calR$, we recover it whenever $\calB$ does not contain any $\Fqk$-subspace.

\medskip

The algorithm we propose to use in the presence of both deletions and insertions at the same time is the following:
\begin{equation}
    \label{eq:Alg1}
    \hat{\calC} = \fred(\fexp(\calR,\bm x),\bm y),
\end{equation}
where $\bm x \in (\Fqk^*)^s, \bm y \in (\Fqk^*)^{k+\varepsilon}$ and $s < k$.
The algorithm works according to the following logic: if only a few deletions have occurred, $\calU$ can be expanded to $a\Fqk$ with a small expansion. Once we obtain a space $\calR_{\mathrm{exp}}$ containing the original codeword $a\Fqk$, we are left with only insertions and can reduce to the case described above. Otherwise, if $a\Fqk$ cannot be reconstructed, the reduction phase will lead to the trivial space $\{0\}$ and the algorithm will fail.
The main challenge of the above process is to find the right expansion length $s.$
In the following we discuss the range of $s$ that can lead to successful decoding.

\smallskip 

In Lemma \ref{lm:exp_to_Fqk_linear} we have seen how, choosing an expanding vector $\bm x$ such that $\Span{\Fq}{x_1,\ldots,x_{k}} = \Fqk$, we obtain $\fexp(\calV,\bm x) = \Span{\Fqk}{\calV}.$
As this space is $\Fqk$-linear, it will be stable under any reduction.
In the case where $\calR = \calU \oplus \calB$, with $\calB = \Span{\Fq}{b_1, \ldots, b_t}$, the space $\fexp(\calU \oplus \calB, \bm x)$, where $\bm x$ is such that $\Span{\Fq}{x_1,\ldots,x_{k}} = \Fqk$, is generated by $a, b_1, \ldots, b_t$ and therefore has $\Fqk$-dimension at most $t+1$. Under this circumstance, any subspace of $\Fqk$-dimension $1$ of this space could be a valid codeword in the subfield spread code.
This means that we will end up with an exponential list of up to $\frac{q^{(t+1)k}-1}{q^k - 1}= q^{kt} + q^{k(t-1)} + \cdots + 1$ possible codewords.
This example gives us an idea of the size of a list when the algorithm fails and a good motivation to choose a short expansion. 

\medskip

An expansion will be successful if $$\fexp(\calU\oplus\calB,\bm x) =\fexp(\calU,\bm x)+\fexp(\calB,\bm x)= \calC+\fexp(\calB,\bm x).$$ Thanks to Corollary \ref{cor: expected dimension of the expansion}, w.h.p. $\dim(\fexp(\calU,\bm x)) = \min\{k, (k-d)s\}$, from which we derive that the maximum dimension of deletions we can correct is bounded by $d \leq k\left(1- \frac{1}{s}\right)$. 

We have already seen how a large expansion can generate some unwanted $\Fqk$-linear subspace that arise from the insertion space $\calB$. 
Consider the space $\calR_x \coloneqq\calR\cap x \Fqk,$ we define $\hat{b} = \max\{ \dim(\calR_x) \mid x \in \Fqn^*\text{ and }x\Fqk\cap a\Fqk=\{0\}\}$ and let $\hat{\calR}$ be one of these intersections with maximal dimension $\hat{b}.$ The $s$-expansion $\fexp(\hat{\calR},\bm x)$ will have its dimension upper-bounded by $\min\{s\hat{b},k\}.$ 
In order to have both a successful expansion and avoid an exponential list, we should choose $s$ such that
\begin{equation}
\label{eq:ER_min_exp}
\frac{k}{k-d} \leq s < \frac{k}{\hat{b}}.
\end{equation}
As the typical value of $\hat{b}$ is one (see Section \ref{sec:Probability of success for a minimum distance decoder} for more details) we can \textit{safely} use $s \leq k-1$ to avoid the total expansion to an $\Fqk$-linear space which is not $a\Fqk$. The expansion $s$ must also be such that $s \dim(\calR) < n$ to avoid encompassing the whole space.
If the dimension of the expanded space is close to $n$ but does not reach it, the space may contain $\Fqk$-linear subspaces of dimension greater than one. In this case, the algorithm outputs a list rather than a unique solution. Proposition~\ref{prop:Fqk_minimal_space} formalizes this idea and provides an additional bound on the expansion length to avoid the situation described above. \par However, before stating and proving the proposition, we recall a simple fact that will be used repeatedly throughout the paper. Let $\calW_1,\ldots,\calW_h$ be subspaces of $\Fqn$ such that $\dim(\calW_i)=t$ for every $i\in\{1,\ldots,h\}$. Then, by iteratively applying Grassmann's formula, we obtain
\begin{equation}
    \label{eq: Grass iter}
    \dim\left(\bigcap_{i=1}^h\calW_i\right)\geq \dim(\calW_1)+\dim\left(\bigcap_{i=2}^h\calW_i\right)-n\geq \ldots\geq ht-(h-1)n=n-h(n-t). 
\end{equation}
\begin{proposition}
\label{prop:Fqk_minimal_space}
    Let $n=kr$ and $\calV \subseteq \Fqn$ be a subspace of $\Fq$-dimension $n-i\geq k$ for $i<r$.
    There exists an $\Fqk$-linear space $\calW \subseteq \calV$ such that $\dim_{\Fqk}(\calW) \geq r - i$, i.e., of $\Fq$-dimension $k(r-i)$.
\end{proposition}
\begin{proof}
Consider a basis $\{1,\alpha,\ldots,\alpha^{k-1}\}$ of $\Fqk$ and define
$$
    \calW = \calV \cap \alpha \calV \cap \ldots \cap \alpha^{k-1} \calV.
$$
From \eqref{eq: Grass iter} we have that
$\dim(\calW) \geq n - ki = (r-i)k$ and by definition we also have $\calW \subseteq \calV.$
To prove $\mathbb{F}_{q^k}$-linearity, it suffices to show that $\alpha W \subseteq W$, which follows by the same argument as in the proof of Lemma \ref{lm:random_red_to_Fqk_linear}.
\end{proof}
From Inequality (\ref{eq:ER_min_exp}) we have $s \geq \frac{k}{k-d}$.
On the other hand, in light of Equation \eqref{eq: min upper bound on dimension of expansion} and Proposition \ref{prop:Fqk_minimal_space}, when we apply Equation (\ref{eq: dim_exp}), we need to limit the expansion $s$ ensuring that 
\begin{equation}
    \label{eq: condition such that there is no arising of an Fqk linear space in the expansion}
    st + k \leq n-r +1,
\end{equation}
and hence 
\begin{equation}
    \label{eq: max s ER with floor}
    s\leq \left\lfloor\frac{n-r+1-k}{t}\right\rfloor.
\end{equation}
To summarize, we obtain the following upper bound on the expansion length
\begin{equation}
    \label{eq:ER_max_exp}
    s \leq \min\left\{\left\lceil\frac{n-r+1-k}{t}\right\rceil-1,\;k-1\right\},
\end{equation}
along with the following upper bound on the dimension of insertions $t=\dim(\calB)$ that the algorithm can correct
\begin{equation}
    \label{eq:ER_max_t}
    t \leq \left\lceil\frac{n-r+1-k}{s}\right\rceil-1
    \leq
    \left\lceil\frac{n-r+1-k}{k}(k-d)\right\rceil-1.
\end{equation}
The expression in \eqref{eq:ER_max_exp} requires a brief explanation. If $t\nmid (n-r+1-k)$, then
\[
\left\lfloor\frac{n-r+1-k}{t}\right\rfloor
=
\left\lceil\frac{n-r+1-k}{t}\right\rceil-1.
\]
On the other hand, if $t\mid (n-r+1-k)$, then
\[
\left\lceil\frac{n-r+1-k}{t}\right\rceil-1
<
\left\lfloor\frac{n-r+1-k}{t}\right\rfloor.
\]
Therefore, as discussed at the beginning of this section, we choose the shorter expansion length as it is always preferable in order to avoid list decoding.

Finally, to use \eqref{eq:ER_max_exp} in the decoding algorithm, as we always need $k - d \geq 2$ and, since $\dim(\mathcal{R}) = k - d + t$, we have $\dim(\mathcal{R}) \geq t + 2$.  
Then, we can use the upper bound for $s$ in the algorithm as
\begin{equation}
\label{eq:max_exp_ER}
    s \leq \min \left\{\left\lceil \frac{n - r + 1 - k}{\dim(\mathcal{R}) - 2} \right\rceil - 1,\; k - 1 \right\}.
\end{equation}
Note that combining (\ref{eq:ER_min_exp}) and (\ref{eq: max s ER with floor}) we also obtain the following relations between the parameters:
\begin{equation}
    \label{eq:ER_max_t_duplicate}
    t < (k-d)\left (  \left \lfloor \frac{n-r}{k} \right \rfloor - 1\right ).
\end{equation}
\par After the expansion phase, we can iteratively reduce the expanded space $\calR^{(0)}=\calR$ to $\calR^{(i+1)} := \calR^{(i)} \cap x \calR^{(0)}$ for some $x \in \Fqk\setminus\Fq$.
If the expansion was strong enough to completely reconstruct the codeword $a \Fqk$ from $\calU$, the reducing phase will lead to the codeword of dimension $k$ or to an $\Fqk$-linear subspace containing the codeword. On the contrary, if the expansion phase is only able to reduce the insertion error but not to fully reconstruct the entire codeword, then repeating the reduction will eventually lead to the trivial subspace $\{0\}$.
We will refer to this simple algorithm as \textit{Expand and Reduce} (ER), and provide it in Algorithm \ref{alg:ExpandandReduce}. 
\\
\par Note that depending on the previous expansion, there are sporadic bad reduction choices that would not work. Consider the space $\calV$ as the result of some expansion $\fexp(\calR, \bm x) = \calR.\calX$ where $\calX$ is the support of $\bm x$ and consider the reduction through $\bm y = (1,y).$
For $y \in \calX \calX^{-1} = \{x_1 x_2^{-1} \mid x_1, x_2 \in \calX \}$ the reduction will contain some space of the same dimension of $\calR$.
Indeed, we have $y = x_1 x_2^{-1}$ and 
$
    \calV = x_1 \calR + x_2 \calR + \hat{\calV},
$
for some $\hat{\calV}$ while 
$$
    y \calV= x_1^2 x_2^{-1}\calR + x_1 \calR + x_1 x_2^{-1} \hat{\calV}.
$$
When we intersect these two spaces, we have $x_1 \calR \subseteq \calV \cap y \calV.$
If we slightly generalize this example, for $\bm y =(y_1,y_2) = y_1(1,y_1^{-1} y_2)$ the condition to avoid in order to decrease the dimension of the space $\calV$ during the reducing phase becomes $y_1^{-1} y_2 \notin \calX \calX^{-1}$; as $\calX \calX^{-1}$ is closed under inversion, it is equivalent to the condition $y_1 y_2^{-1} \notin \calX \calX^{-1}$.
Notice that, since $\calX \calX^{-1}$ is relatively small (see Lemma \ref{lm:XX_inv}), the probability of randomly choosing a bad reduction is small.
Moreover, such cases are not problematic in practice, as one can simply repeat the reduction procedure with a different $y$ until the desired dimension is reached.
\begin{algorithm}[t]
            \DontPrintSemicolon
            \caption{\textit{Expand and Reduce} (ER)}
            \label{alg:ExpandandReduce}
            \KwInput{Received subspace $\calR$}
            Set $\mathrm{MaxExp}=\min\left\{\left\lceil\frac{n-r+1-k}{\dim(\calR)-2}\right\rceil-1,k-1\right\}$\;
            \If{$\mathrm{MaxExp}>0$}
                {
                    Choose $x_1,\ldots,x_{\mathrm{MaxExp}}\in\Fqk^*$\;
                    $\calR_{\mathrm{exp}}\coloneqq x_1\calR+\cdots+x_{\mathrm{MaxExp}}\calR$\;
                    $\mathrm{Res}\coloneqq\calR_{\mathrm{exp}}$\;
                    \While{$\dim(\mathrm{Res})>0$\textcolor{blue}{\footnotemark[1]}}
                {Choose $z\in\Fqk\setminus\Fq$\;
                $\mathrm{Res}\coloneqq \mathrm{Res}\cap z\calR_{\mathrm{exp}}$\;
                    \If{$\dim(\mathrm{Res})=k$}
                    {\KwOutput{Estimated codeword $\mathrm{Res}$}}
                } \If{$\dim(\mathrm{Res})\neq k$}
                    {\KwOutput{Decoding failure}}
                }
            \Else {\KwOutput{Decoding failure}}
                
    \end{algorithm}
\footnotetext[1]{\label{fn: footnote of alg 1} Note that, in the case of the emergence of an $\Fqk$-linear space in $\mathrm{Res}$ of dimension strictly greater than one, the algorithm may enter an infinite loop during the iterative process in Step~6. To prevent this issue, in the source code available at the following \href{https://github.com/ermes1990/SbfieldSpreadDecoding.git}{GitHub link}, we impose an upper bound on the dimension of iterations, set to a value of the form $k+\calO(1)$ (e.g., $k+5$ in the implementation). Indeed, under favorable conditions, namely, when Steps~$3$-$4$ successfully reconstruct the entire codeword $a\Fqk$, we have $\mathrm{Res} = a\Fqk \oplus \tilde{\calB}$, where $\tilde{b} \coloneqq \dim(\tilde{\calB}) \leq st$. In this case, the iterative process typically converges to $a\Fqk$ within $k$ iterations.}
\bigskip
\\
\textbf{Success Probability}.
Let $\calR =\calU\oplus\calB$ where $\calU$ is an $\Fq$-subspace of $a \Fqk$ of dimension $k-d$ and $\calB$ is a subspace of dimension $t$ chosen uniformly at random from all the possible $\Fq$-linear subspaces of $\Fqn.$
As we will explain at the beginning of Section \ref{sec:Probability of success for a minimum distance decoder}, one case in which the algorithm fails occurs when for some $x \in \Fqn^*$ such that $x \Fqk \neq a \Fqk$, $\calR_x = \calR \cap x \Fqk$ has dimension $\eta\geq k - d = \dim(\calU)$. In this case, the expansion step is equally effective on $\calR_x$ as it is on $\calU$, which may lead to a list of outputs or to a wrong result. However, as we will show in Section \ref{sec:Probability of success for a minimum distance decoder}, this situation is rather unlikely to occur, even when $\dim(\calU) = 2$.
\par The main cause of failure of this algorithm is that the expansion phase does not fully reconstruct the entire codeword before the reduction phase. The success probability of Alg. \ref{alg:ExpandandReduce}, i.e., the probability of full expansion, is estimated by applying the following proposition.
\begin{proposition}
    \label{prop: expected dimension of the expansion when alpha x beta geq n}
    Let $\calB$ be a fixed $\Fq$-subspace of $\Fqn$ of dimension $\beta$ and let $\calA=\langle a_1,\ldots,a_{\alpha}\rangle_{\Fq}$ where $a_1,\ldots,a_{\alpha}$ are $\alpha$ $\Fq$-linearly independent elements of $\Fqn$ chosen uniformly at random.  If $\dim(\calA)\dim(\calB)\geq n$, then $\calA.\calB=\Fqn$ with probability at least $1-q^{-(\alpha\beta-n)}$.
\end{proposition}
\begin{proof}
    Let us define
    \[
    \calA^{\upmodels}\coloneqq\{x\in\Fqn\colon\mathrm{Tr}_{q^n/q}(xa)=0\text{ for all }a\in\calA\}.
    \]
    From \cite[\S 17, Theorems 1 and 2]{halmos1958finite} and \cite[Theorem 2.24]{lidl1994introduction}, $\calA.\calB\subseteq\Fqn$ is not the whole space if and only if there exists $y\in\Fqn^*$ such that $\mathrm{Tr}_{q^n/q}(y(ab))=0$ for all $a\in\calA, b\in\calB$. 
    If such $y$ exists, we also have that for all $b\in\calB, a\in\calA$, $\mathrm{Tr}_{q^n/q}((yb)a)=0$ , i.e., $y\calB\subseteq \calA^{\upmodels}$. On the other hand, thanks to \cite{taylor1992geometry}, $\delta\coloneqq\dim(\calA^{\upmodels})=n-\dim(\calA).$ For a fixed $\beta$-dimensional subspace $y\calB$, the probability that a random $\delta$-dimensional subspace $\calA^{\upmodels}$ contains $y\calB$ is \[
        \frac{\gbinom{n-\beta}{\delta-\beta}_{q}}{\gbinom{n}{\delta}_{q}}=\frac{\gbinom{n-\beta}{\alpha}_{q}}{\gbinom{n}{\alpha}_{q}} \leq q^{-\alpha\beta},
    \] 
    where the last inequality follows from $\prod_{i=0}^{\alpha-1} (1-q^{-n+i}) \geq \prod_{i=0}^{\alpha-1} (1-q^{-n+\beta+i}).$
    Consequently, 
    \begin{align*}
        \mathrm{P}(\calA.\calX\subsetneq\Fqn)&=\mathrm{P}(\text{There exists }y\in\Fqn^*\colon y\calB\subseteq \calA^{\upmodels})\\
        &\leq \frac{q^n-1}{q-1}\frac{\gbinom{n-\beta}{\delta-\beta}_{q}}{\gbinom{n}{\delta}_{q}}\leq q^{-(\alpha\beta-n)}.
    \end{align*}
    The claim is an immediate consequence of considering the complements of the events. 
\end{proof}
\begin{corollary}
\label{cor:expected_dimension_without_independence}
Let $\calB$ be a fixed $\Fq$-subspace of $\Fqn$ of dimension $\beta$ and let $\calA=\langle a_1,\ldots,a_\alpha\rangle_{\Fq},$ where $a_1,\ldots,a_\alpha$ are chosen uniformly at random from $\Fqn^*$. If $\alpha\beta\ge n$, then
\[
\Pr(\calA.\calB=\Fqn)
\ge
\left(1-q^{-(\alpha\beta-n)}\right)\prod_{i=n-\alpha+1}^{n}\left(1-q^{-i}\right).
\]
\end{corollary}

\begin{proof}
The proof follows immediately from the fact that
\[
\Pr(\calA.\calB=\Fqn)=\Pr(\dim(\calA)=\alpha)\Pr(\calA.\calB=\Fqn\mid \dim(\calA)=\alpha).
\]
\end{proof}
When considering the explicit process of expansion, according to corollary \ref{cor:expected_dimension_without_independence} we see that,  if $\dim(\calU)\dim(\calX)\geq k$, then the success probability of ER converges to one as the field size increases. 
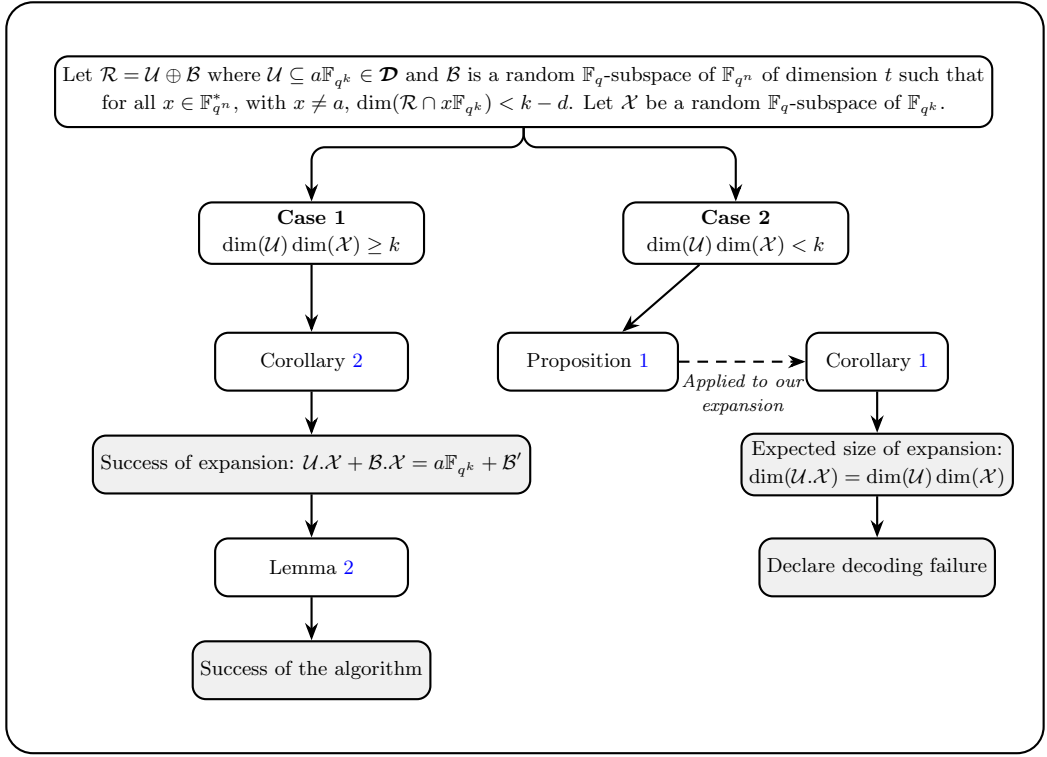
\begin{figure}[h]
    \begin{center}
        \begin{tikzpicture}[
    scale=0.85, transform shape,
    block/.style={
        draw=black, thick, rectangle, 
        minimum width=2.8cm, minimum height=0.9cm, 
        rounded corners=5pt, align=center, fill=white, font=\small
    },
    title/.style={
        fill=white, text=black, font=\large\bfseries, inner sep=5pt
    },
    lightgrayblock/.style={
        block, fill=gray!12
    },
    arrow/.style={
        ->, >=Stealth, thick, rounded corners=5pt
    },
    dashedarrow/.style={
        ->, >=Stealth, thick, dash pattern=on 4pt off 3pt, rounded corners=5pt
    },
    container/.style={
        draw=black, rectangle, rounded corners=10pt, thick, inner sep=18pt
    }
]

\node[block, minimum width=8.2cm, minimum height=1.1cm] (eq1) at (1.5, 4.0) {Let $\calR=\calU\oplus\calB$ where  $\calU\subseteq a\Fqk\in\bm\calD$ and $\calB$ is a random $\Fq$-subspace of $\Fqn$ of dimension $t$ such that\\ for all $x\in\Fqn^*$, with $x\neq a$, $\dim(\calR\cap x\Fqk)<k-d$. Let $\calX$ be a random $\Fq$-subspace of $\Fqk$.};

\node[block, minimum width=3.5cm] (case1) at (-1.8, 1.8) {\textbf{Case 1}\\$\dim(\calU) \dim(\calX) \ge k$};
\node[block, minimum width=3.5cm] (case2) at (4.8, 1.8) {\textbf{Case 2}\\$\dim(\calU) \dim(\calX) < k$};

\node[block, minimum width=3.0cm] (prop4) at (-1.8, -0.2) {Corollary \ref{cor:expected_dimension_without_independence}};
\node[lightgrayblock, minimum width=3.0cm] (succ) at (-1.8, -1.8) 
{Success of expansion: $\calU.\calX+\calB.\calX=a\Fqk+\calB^{\prime}$};

\node[block, minimum width=3.0cm] (prop2) at (-1.8, -3.4) {Lemma \ref{lm:random_red_to_Fqk_linear}};
\node[lightgrayblock, minimum width=3.0cm] (successq) at (-1.8, -5.0) {Success of the algorithm};

\node[block, minimum width=2.8cm] (prop1) at (2.5, -0.2) {Proposition \ref{prop: expected dimension of the expansion}};
\node[block, minimum width=2.2cm] (cor1) at (7.0, -0.2) {Corollary \ref{cor: expected dimension of the expansion}};
\node[lightgrayblock, minimum width=2.4cm] (exp_size) at (7.0, -1.8) 
{Expected size of expansion:\\
$\dim(\calU.\calX)=\dim(\calU)\dim(\calX)$};

\node[lightgrayblock, minimum width=2.4cm] (dec_fail) at (7.0, -3.4)
{Declare decoding failure};

\draw[arrow] (eq1.south) -- ++(0,-0.3) -| (case1.north);
\draw[arrow] (eq1.south) -- ++(0,-0.3) -| (case2.north);

\draw[arrow] (case1) -- (prop4);
\draw[arrow] (prop4) -- (succ);

\draw[arrow] (succ) -- (prop2);
\draw[arrow] (prop2) -- (successq);

\draw[arrow] (case2) -- (prop1);
\draw[arrow] (cor1) -- (exp_size);
\draw[arrow] (exp_size) -- (dec_fail);

\draw[dashedarrow] (prop1.east) -- node[below, font=\footnotesize\itshape, align=center, yshift=-2pt] {Applied to our\\expansion} (cor1.west);

\node[container, fit={(eq1) (case1) (succ) (prop2) (successq) (cor1) (dec_fail)}, inner ysep=20pt] (box1) {};

\end{tikzpicture}
    \end{center}
    \caption{Overview of the ER Algorithm}
    \label{fig: ER Algorithm Process Overview}
\end{figure}

\textbf{Complexity of Alg. \ref{alg:ExpandandReduce}.}
Expansion complexity is equivalent to extract a basis from the basis of each $x_i \calR$.
That is the same of performing Gaussian reduction of a system of $sr$ rows and $n$ columns, where $s$ is the dimension of expansions and $r= \dim(\calR).$
This has complexity in the order of $O(srn)$ as $s < k <n$ and $r < n$ it is dominated by $O(n^3)$.

The reduction is the intersection of $k$ spaces of dimension $r$ in $\Fqn$.
Intersecting two spaces of dimension $r$ has the cost of reducing a $2r \times n$ matrix which is dominated by $O(n^3)$, this operation will be repeated $k-1$ times, depending on the relation between $k$ and $n$, the total complexity expressed in terms of $n$ will be between $O(n^3)$ and $O(n^4)$

\subsection{Application of the Algorithm to $\mathrm{Gr}_{\Fqk}(i,\Fqn)$.}
Before concluding the discussion about decoding, it is worth noting that the proposed algorithm can be used to decode other types of subspace codes closely related to subfield spreads codes. 
\par Let $n=kr$, the subfield spread associated with the intermediate field $\Fqk \subseteq \Fqn$ can be also seen as $\Gr_{\Fqk}(1,\Fqn) \subseteq \Gr_{\Fq}(k,\Fqn) $ which is the collection of all the $\Fqk$-linear subspaces of $\Fqn$ of $\Fqk$-dimension $1$.
It is straightforward to extend the decoding algorithm to the constant-dimension codes $\Gr_{\Fqk}(i,\Fqn) \subseteq \Gr_{\Fq}(ik,\Fqn)$ for which $ik \leq n.$
These codes have constant dimension $ik$ and it is easy to show that they still have minimum subspace distance $2k$.
The maximal cardinality will be achieved for $i\in\{\lfloor\frac{n}{2}\rfloor, \lceil\frac{n}{2}\rceil\}$ and is equal to 
\[
\gbinom{n}{i}_{q^k} = \prod_{j=0}^{i-1} \frac{q^n - q^{jk}}{q^{ik}-q^{jk}}.
\]
Let $\calR$ be the received space. By Lemma~\ref{lm:exp_to_Fqk_linear}, the operator $\fexp$ yields the smallest $\Fqk$-linear subspace containing $\calR$ for some choices of the expanding vector.
As previously observed, when $\calR$ contains a sufficiently large portion of an $\Fqk$-linear subspace, a small number of expansions may already recover this space entirely.
Moreover, by Proposition~\ref{prop:Fqk_minimal_space}, the operator $\fred$ extracts the smallest $\Fqk$-linear subspace. In many cases, this is sufficient to recover the original codeword. Notice that a large value of $i$ typically means a larger received space $\calR$, this limits the number of expansions and negatively affects the failure rate.

As a final remark, the code of minimum subspace distance $k$ given by the balls 
\[
\calB_{\Fqk}(i,\Fqn) = \bigcup_{j=1}^i \Gr_{\Fqk}(j, \Fqn)
\] 
can be decoded with similar techniques.
\section{Improved Techniques: Expand Reduce Expand and Filtered ERE}
\label{sec: Improved techniques: Expand Reduce Expand and Filtered ERE}
\subsection{Expand Reduce Expand (ERE)}
Observe  that, if we stop the reduction one step before reaching the space $\{0\}$ it is likely that this space is just a subspace of $a\Fqk$.
At this point an expansion in the order of $k$ will reconstruct the original codeword. In this way we can improve the previous algorithm.
We will refer to this second algorithm as \textit{Expand Reduce Expand} (ERE) and we provide it in Algorithm \ref{alg:ExpandReduceExpand}.
To justify this claim,  consider the case when $\calU$ is expanded to a space $\hat{\calU}$ of dimension $k - \varepsilon$ for some small integer $\varepsilon\leq d.$
\par Since for any $x \in \Fqk \setminus \Fq$, $\hat{\calU}+x \hat{\calU}\subseteq a\Fqk$, thanks to the Grassmann's formula, we have that 
$\dim(\hat{\calU} \cap x \hat{\calU}) \geq \max\{k - 2 \varepsilon,0\}$. More generally, by Equation \eqref{eq: Grass iter}, we get $\dim(\fred(\hat{\calU}, \bm y)) \geq \max\{k - m \varepsilon,0\}$ where $\bm y \in \Fqk^m$.
This means we can use a reduction of length up to $m \leq   \lfloor \frac{k}{\epsilon} \rfloor$ before losing any trace of $\calU$. Similarly, for a random space $\calW \subseteq \Fqn$ of dimension $\dim(\calW) = n - c$, its reduction has dimension 
$\dim(\fred(\calW, \bm y)) \geq \max\{n - mc,0\}$, where the equality holds with high probability thanks to Remark \ref{prop: probability reduction goes to zero}.
This space will be reduced to $\{0\}$ after $m = \lceil \frac{n}{c} \rceil$ intersections. \par
In our case, after receiving an $\mathbb{F}_q$-subspace $\calR$ containing $\calU$ and applying an $s$-expansion with $\bm x\in\mathbb{F}_{q^k}^s$, the expected dimension of $\fexp(\calR,\bm x)$ and $\fexp(\calU,\bm x)$ is $s\dim(\calR)=n-(n-\dim(\calR)s)$ and $s\dim(\calU)=k-(k-\dim(\calU)s)=k-(k-(k-d)s)$, respectively. 
Therefore, for an $s$-expansion, the expected value of $c$ and $\varepsilon$ is $c=n-\dim(\calR)s$ and $\varepsilon=k - (k-d)s$ respectively.
Consequently, if $ \frac{n}{c} <m< \frac{k}{\varepsilon}$, we can expect the intersection to give a space that is completely contained in $a \Fqk.$
We can rewrite the inequality as
$$
 \frac{n}{n-\dim(\calR)s} < m < \frac{k}{k - (k-d)s} 
$$
from which, recalling $n = kr$ and $\dim(\calR)=k-d+t$, we obtain the following upper bound on the dimension of insertion
\begin{equation}
\label{eq:ERE_max_t}
    t < (k-d)(r-1).
\end{equation}

For both ER and ERE the result depends on the choice of $\bm x \in \Fqk^s$ we use to expand and $\bm y \in \Fqk^m$ we used to reduce.
A failure is either a subspace of dimension different from $k$ or a wrong subspace of dimension $k$.
In the second case, there is little we can do, but in the first case, which is more common by experimental observations, we can repeat the experiment until we get some space of dimension $k$. We terminate the algorithm with a failure after we reach a maximal numbers of trials to avoid long computations.
The complexity of a single ERE execution is the same as the complexity of a single ER execution.

In ERE there is no need for the first expansion to fully reconstruct the original codeword, then the condition $(k-d)s \geq k$ is no longer necessary. In particular, since w.h.p. $\dim(\fexp(\calR, \bm x)) = s \dim(\calR)$, Proposition~\ref{prop:Fqk_minimal_space} shows that we just need
\[
s \dim(\calR) \leq n-r+1,
\]
and hence
\begin{equation}
    \label{eq:max_exp_ERE}
    s \leq \min \left\{\left \lceil \frac{n-r+1}{\dim(\calR)} \right \rceil - 1 , k-1 \right\}\footnote{Once again, we prefer $\left\lceil \frac{n-r+1}{\dim(\calR)} \right\rceil - 1$ over $\left\lfloor \frac{n-r+1}{\dim(\calR)} \right\rfloor,$ for the reasons discussed in connection with Equation~\eqref{eq:ER_max_exp}.}.
\end{equation}
\begin{algorithm}[htbp]
            \DontPrintSemicolon
            \caption{\textit{Expand Reduce Expand} (ERE)}
            \label{alg:ExpandReduceExpand}
            \KwInput{Received subspace $\calR$}
            Set $\mathrm{MaxExp}=\min\left\{\left\lceil \frac{n-r+1}{\dim(\calR)} \right \rceil-1,k-1\right\}$\textcolor{blue}{\footnotemark[2]}\;
            \If{$\mathrm{MaxExp}>0$}
                {
                    Choose $x_1,\ldots,x_{\mathrm{MaxExp}}\in\Fqk^*$\;
                    $\calR_{\mathrm{exp}}\coloneqq x_1\calR+\cdots+x_{\mathrm{MaxExp}}\calR$\;
                    $\mathrm{Res}:=\calR_{\mathrm{exp}}$\; 
                    \While{$\dim(\mathrm{Res})>0$\textcolor{blue}{\footnotemark[3]}}
                {$\mathrm{Res}_{\mathrm{old}}:=\mathrm{Res}$\;
                Choose $z\in\Fqk\setminus\Fq$\;
                $\mathrm{Res}:=\mathrm{Res}\cap z\calR_{\mathrm{exp}}$\;
                }$\mathrm{Res}:=\mathrm{Res}_{\mathrm{old}}$\;
            \While{$\dim(\mathrm{Res})<k$}
                {Choose $x\in\Fqk\setminus\Fq$\;
                    $\mathrm{Res}:=\mathrm{Res}+x\mathrm{Res}$}
            \KwOutput{Estimated codeword $\mathrm{Res}$}
                }
                \Else{\KwOutput{Decoding failure}}
    \end{algorithm}
    \footnotetext[3]{More discussions in
    Section \ref{sec:experimental_results}.}
    \footnotetext[2]{Footnote \ref{fn: footnote of alg 1} also applies in this case.}
\subsection{Filtered ERE}
\label{subsec:filtered ERE}
We have seen how Equation (\ref{eq:ERE_max_t}) limits the dimension of insertions we can tolerate.
For $\bm x \in (\Fqk^*)^s$ such that $\Span{\Fq}{\bm x} = \calX$ is of dimension $s$ and $y \in \Fqk$ such that $y\notin \calX$, the space $\calX + \Span{\Fq}{y}$ is a $s+1$-dimensional space, we denote by $\calX + \calY$.
If $\dim(\fexp(\calR, \bm x)) = \dim(\calR.\calX) = s \dim(\calR),$ the space $\calX.\calR$ and $y \calR$ have a nontrivial intersection if and only if 
$$
\dim((\calX + \calY).\calR) < (s+1)\dim(\calR).
$$
Similarly, if $\dim(\fexp(\calU,\bm x))=s(k-d),$ the intersection of $\calX.\calU$ and $y \calU$ is not trivial if and only if $\dim((\calX+\calY).\calU) < (s+1)(k-d).$
According to the proofs of Proposition \ref{prop: expected dimension of the expansion} and Corollary \ref{cor: expected dimension of the expansion}, the probabilities of having a nontrivial intersection in the two cases can be estimated to be on the order of $q^{(s+1) \dim(R) - n}$ and $q^{(k-d)(s+1) - k}$, respectively. 

The second probability is larger than the first if
$
    -n + (s+1)\dim(\calR) < - k+ (s+1)(k-d)
$,
which gives the condition on $s$ as follows:
\begin{equation}
    \label{eq:filter_cond}
    s < \frac{n-k}{\dim(\calR)-(k-d)} - 1 = \frac{n-k}{t} - 1.
\end{equation}
We still need to avoid the expansion to generate extra $\Fqk$-linear spaces. So the number of expansions $s$ still has to satisfy Inequality \eqref{eq:max_exp_ERE}.
Notice that since $y\calU \cap \hat{\calU} \subseteq y\calR \cap \hat{\calR}$, if the second intersection is nontrivial and the Condition (\ref{eq:filter_cond}) is satisfied, it is likely that this intersection belongs to the original codeword $a \Fqk$.
The idea of filtering is to try different combinations of $\bm x$ and $y$ and repeat the intersection $\fexp(\calR,\bm x) \cap y \calR.$
\par Starting from $\calF^{(0)} = \{0\}$ we can iteratively update it as
$$
    \calF^{(i+1)} = \calF^{(i)} + \fexp(\calR, \bm x^{(i)}) \cap y^{(i)} \calR,
$$
where the sum is the sum of subspaces, until we achieve the desired threshold dimension $T$.
We refer to this process as filtering, as we start from a large space $\calR$ containing many insertions compared to $\calU$ and, during the filtering, we collect small subspaces that are more likely to be contained in the original codeword than being originated by the insertions.
Let $\calF^{(n)}$ be the space with the target dimension obtained after this filtration.
We cannot expect the filter to work perfectly; that is, we do not expect $\calF^{(n)}$ to be the original codeword. However, we do expect $\calF^{(n)}$ to have a larger intersection with the original codeword than $\calR$ does, thereby increasing the chances that ERE succeeds when applied to $\calF^{(n)}$ rather than to $\calR$.
\par A good heuristic to estimate the probability of each filtered element belonging to the correct subspace $a \Fqk$ can be given in the following way. 
Let $A$ be the event $x \in \calX.\calU \cap y \calU$ and $B$ be the event $x \in \calX.\calR \cap y \calR$. 
Denote by $\pr_A$ and $\pr_B$ the corresponding probabilities $\pr_A, \pr_B$. Under the assumption that the probability of $\calX.\calU \cap y\calU \neq \{0\}$ is greater than the probability of $\calX.\calR \cap y\calR \neq \{0\}$, i.e., when Condition (\ref{eq:filter_cond}) is satisfied, we have $\pr_A \geq \pr_B.$ 
There is an implicit third event of probability, i.e., $1 - \pr_A - \pr_B$, that corresponds to obtaining a trivial intersection.
Let $x \neq 0$ be an element from a nontrivial intersection. Since $\calX.\calU \cap y \calU \subseteq a \Fqk$, we know that, if $x$ is observed thanks to event $A$, we automatically have $x \in a \Fqk$.
Applying Bayes' Theorem, we obtain the following bound on the probability of $x \in a \Fqk$:
\begin{equation}
\label{eq:condition probability for filter}
\pr(x \in a \Fqk) \geq \frac{\pr_A}{\pr_A + \pr_B} =
1 - \frac{\pr_B}{\pr_A+\pr_B} \approx 1-\frac{\pr_B}{\pr_A} = 1 - q^{(s+1)t - (n-k)}.
\end{equation}
At this point we can use ERE to obtain the original codeword.
Basically, the filtering is useful to reduce the dimension of insertion and obtain a space whose dimension of the insertion is small enough to apply ERE.
Notice that, from Condition (\ref{eq:filter_cond}) we can find a limit to the dimension of insertion that the filtering method can handle, which is given by
\begin{equation}
    \label{eq:filter_ins_limit}
    t < \frac{n-k}{s+1}.
\end{equation}
As the smallest number of expansions we can use is $1$, this becomes $t < \frac{n-k}{2}$. Experimental results in Section \ref{sec:experimental_results} show that the performance starts to degrade around $t \geq \frac{n-k}{2} - 2$. 
\par Let $\pi_a\colon\Fqn\to\Fqn/a\Fqk$ be the canonical quotient map. In Section \ref{sec:Probability of success for a minimum distance decoder}, we show that the success of our algorithm strongly depends on the property of $\pi_a(\mathcal{B})$ of being $\eta$-evasive with respect to the Desarguesian spread
\[
\pi_a(\bm{\calD})=\{\pi_a(b\Fqk) \colon b\in\Fqn^*\text{ and } b\Fqk\cap a\Fqk=\{0\}\}
\]
in $\Fqn/a\Fqk$, for some $1\leq \eta<k-d$.
Moreover we also show that if $\calB$ is an $\Fq$-subspace of dimension $t$ chosen uniformly at random in $\Fqn$ such that $\calB\cap a\Fqk=\{0\},$ then $\pi_a(\mathcal{B})$ is uniformly distributed among the $t$-dimensional $\F_q$-subspaces of the quotient space $\Fqn/a\Fqk$. 
In \cite[Remark 5.5]{gruica2024generalised} it is shown that, as $q \to \infty$ the probability of $\pi_a(\calB)$ of dimension $t=\dim(\calB)$ being $\eta$-evasive with respect to the Desarguesian spread $\pi_a(\bm{\calD})$ is $1$ for $t \leq \frac{\eta}{\eta+1}(n-2k)+\eta$ and drops to zero for $t \geq\frac{\eta}{\eta+1}(n-2k)+\eta+ 2.$

In our experiments, we consider the case of $k-2$ deletions which can be decoded only if the projection of the insertion has evasiveness $\eta = 1$ (see Section \ref{sec:Probability of success for a minimum distance decoder} for further explanation).
In this case, the condition on $t$ then becomes $t \leq \frac{n-2k}{2} + 1$.

In \cite[Figure 2]{gruica2024generalised} it was already noticed that, for $q=2$, the threshold seems to become smoother and a decrease in the proportion of $\eta$-evasive subspaces is observed starting around $t \geq \frac{n-2k}{2} - 1$ which is consistent with what we observe in our experiments (see Section \ref{sec:experimental_results}, Figures \ref{fig:88_ERE_graph} and \ref{fig:65_graph}). 
\\
\par Using a smaller value of $s$ generally leads to a higher-quality filtered space, i.e., fewer insertions. However, it also requires performing more filtering iterations before obtaining a filtered space with the target dimension.
Assuming that each intersection gives one element in $a\Fqk$ with probability $q^{(s+1)(k-d)-k}$, we will need on average $T q^{k - (s+1)(k-d)}$ attempts before reaching the threshold dimension $T$ for the filtered space.
Treating $T$ as a constant, the total complexity of filtered ERE is given by $O(sn^{3} q^{k-(s+1)(k-d)})$, as in the analysis of ER, the parameter $s$ can be considered to be between a constant or linear in $n$ depending on the relation between $r$ and $k$. 
Notice that the total complexity is still polynomial in $n$ only if $k$ is considered constant, if we choose for example $k = \sqrt{n}$ for the extreme regime $s=1, k-d = 2 $ the complexity will be sub-exponential in the order of $O(n^3 q^{\sqrt{n} - 4})$.
This is still a much better complexity than the naive approach of testing all possible $\frac{q^n - 1}{q^k - 1}$ codewords.

Notice that it is better to set $T \leq k$.
The main reason is that, otherwise, when assuming Condition (\ref{eq:filter_cond}), after we recover all the $k$ elements generating the original codeword, new vectors cannot come from that space anymore meaning we will need on average many more iterations, moreover the filtered space we obtain will always contain insertions contrary to the case when $T< k$ where we can expect to get an insertion free space.

\par We will refer to this third algorithm as \textit{Filtered} ERE and we describe it in Algorithm \ref{alg:FilteredERE}. 
\begin{algorithm}[htbp]
            \DontPrintSemicolon
            \caption{\textit{Filtered} ERE}
            \label{alg:FilteredERE}
            \KwInput{Received subspace $\calR$, threshold dimension $T \leq k$ \;}
            $\calF:=\{0\}$\;
            $\mathrm{MaxExp}\coloneqq\min\left\{\left\lceil \frac{n - r + 1}{\dim(\mathcal{R})-2} \right\rceil - 2, \left\lfloor \frac{n-k}{\dim(\calR)-2} \right\rfloor - 1 ,k-1\right\}$\textcolor{blue}{\footnotemark[4]}\;
            \While{$\dim(\calF)<T$}
                {
                    Choose $y\in\Fqk^*$\;
                    $\calR_1 := y\calR$ \;
                    \If{$\mathrm{MaxExp}>0$}
                    {
                        Choose $x_1,\ldots,x_{\mathrm{MaxExp}}\in\Fqk^*$\;
                        $\calR_{2}:=x_1\calR_2+\cdots+x_{\mathrm{MaxExp}}\calR_2$
                    }
                    $\mathrm{Extract} :=\calR_1\cap \calR_2$\;
                    $\calF:=\calF+\mathrm{Extract}$\;
                }
            Run Algorithm \ref{alg:ExpandReduceExpand} on the input $\calF$\;
            \KwOutput{Estimated codeword}
    \end{algorithm}
\footnotetext[4]{See section \ref{sec:experimental_results}.}

To conclude this section, we provide a summary table showing the dimension of random insertions handled by the three algorithms. 
\begin{table}[!htbp]
\centering
\begin{tabular}{c|c}
\hline 
 Decoding alg. & \textbf{Dimension of insertions handled} \\
\hline
ER  & $<(k-d)\left(\left\lfloor\frac{n-k}{k}\right\rfloor-1 \right)$ \\
ERE & $<(k-d)(r-1)$ \\
Filtered ERE & $<\frac{n-k}{2}$ \\
\hline 
\end{tabular}
\caption{Dimension of random insertions handled by the three algorithms. Note that the bound for Filtered ERE does not depend on the dimension of the deletion.}
\label{tab:insertions-handled}
\end{table}

\section{Experimental Results}
\label{sec:experimental_results}
We implemented the three algorithms presented in the previous section, namely, \textit{Expand and Reduce} (ER), \textit{Expand Reduce Expand} (ERE) and the \textit{Filtered ERE}.
The experiments were implemented in SageMath and the source code is available at this \href{https://github.com/ermes1990/SbfieldSpreadDecoding.git}{GitHub link}\footnote{\url{https://github.com/ermes1990/SbfieldSpreadDecoding}}.
Results are summarized in Figures \ref{fig:88_graph} and \ref{fig:65_graph}.

In both experiments, we generate $1000$ corrupted codewords and try to decode each with the three algorithms presented in the paper.
For all the experiments, deletions are fixed to $k-2$ which is the maximum we can expect to correct.
With less deletions the algorithm would improve both the speed and the accuracy of the experiments.
The size of the first expansion in ER is the maximal expansion we calculate in (\ref{eq:max_exp_ER}) and it is given by
$$s \leq\min \left \{\left\lceil \frac{n-r+1-k}{\dim(\calR)-2} \right \rceil - 1,k-1 \right\},$$
where $n=kr$ is the dimension of the ambient space and $\calR$ is the received subspace. 

During the experiments with ERE, instead, we observed that the maximal expansion calculated in \eqref{eq:max_exp_ERE} was not always the optimal choice.
In particular, we found that reducing the expansion length by one often led to an improvement in the success rate.
Motivated by this observation, in the ERE algorithm we instead use the slightly smaller expansion
\[
\min\left\{\left\lceil \frac{n-r+1}{\dim(\mathcal{R})} \right\rceil - 2,\; k-1\right\}.
\]

For the Filter in filtered ERE we used the minimum between the Expansion
(\ref{eq:filter_cond}) combined with (\ref{eq:max_exp_ERE}) and the expansion used by ERE to maintain a fair comparison between the two.
In particular we used the expansion
 $\min\left\{\left\lceil \frac{n - r + 1}{\dim(\mathcal{R})-2} \right\rceil - 2, \left\lfloor \frac{n-k}{\dim(\calR)-2} \right\rfloor - 1 ,k-1\right\}.$
Due to the parameters involved in the experiments, we repeat the execution of ERE and ER a maximum of $5$ times until we get an output of dimension $k$.
The reason we stop at $5$ is to avoid long computations.
In Filtered ERE we do the filtration only one time with threshold $5$ as it can be very expensive, after the filtration, we run ERE a maximum of $5$ times.

\begin{figure}[H]
\begin{tikzpicture}
\begin{axis}[
    width=12cm,
    height=8cm,
    xlabel={insertions},
    ylabel={Success rate (\%)},
    y filter/.code = \pgfmathparse{#1/10},
    xmin=1, xmax=31,
    ymin=0, ymax=100,
    grid=major,
    legend pos=south west,
    legend style={
        font=\scriptsize,
        row sep=1pt
    }
]

\addplot[
    thick,
    red,
    mark=x
] table [col sep=comma, x=insertions, y=alg1] {test_8_8.csv};
\addlegendentry{ER}

\addplot[
    thick,
    blue,
    mark=*
] table [col sep=comma, x=insertions, y=alg2] {test_8_8.csv};
\addlegendentry{ERE}

\addplot[
    thick,
    green,
    mark=square*
] table [col sep=comma, x=insertions, y=alg3] {test_8_8.csv};
\addlegendentry{Filtered ERE}

\end{axis}
\end{tikzpicture}
\caption{
Success rates of decoding algorithms for parameters $q=2$, $n=64$ and $k=8$.
}
\label{fig:88_graph}
\end{figure}
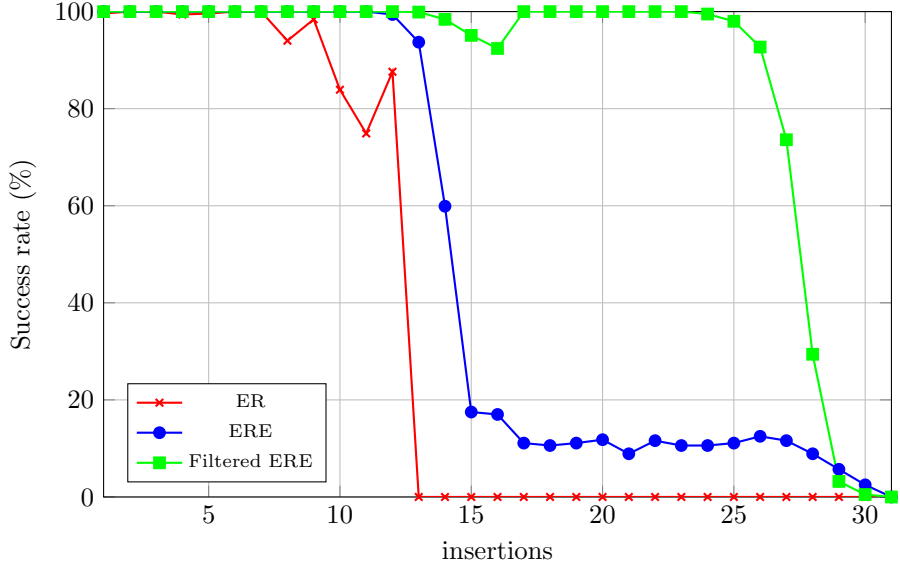

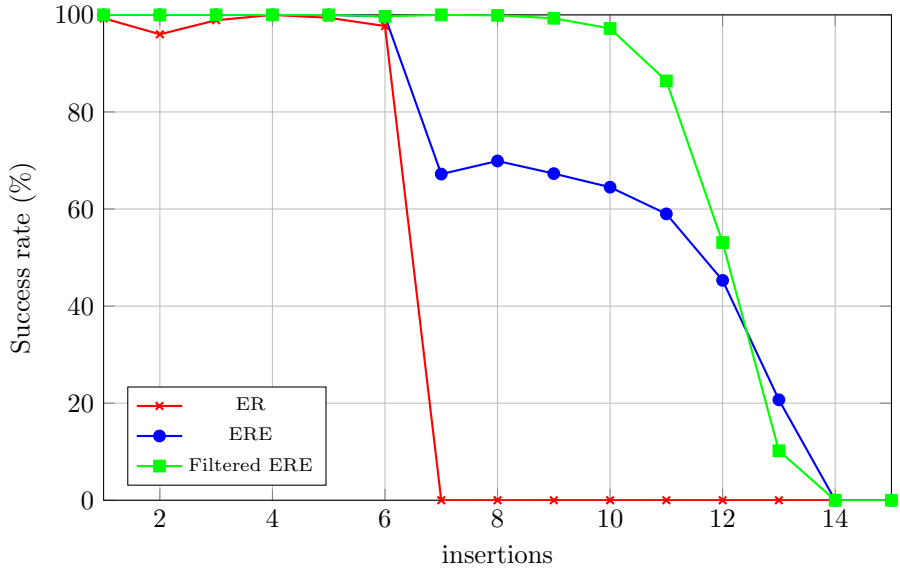
\begin{figure}[H]
\begin{tikzpicture}
\begin{axis}[
    width=12cm,
    height=8cm,
    xlabel={insertions},
    ylabel={Success rate (\%)},
    y filter/.code = \pgfmathparse{#1/10},
    xmin=1, xmax=15,
    ymin=0, ymax=100,
    grid=major,
    legend pos=south west,
    legend style={
        font=\scriptsize,
        row sep=1pt
    }
]

\addplot[
    thick,
    red,
    mark=x
] table [col sep=comma, x=insertions, y=alg1] {test_6_5.csv};
\addlegendentry{ER}

\addplot[
    thick,
    blue,
    mark=*
] table [col sep=comma, x=insertions, y=alg2] {test_6_5.csv};
\addlegendentry{ERE}

\addplot[
    thick,
    green,
    mark=square*
] table [col sep=comma, x=insertions, y=alg3] {test_6_5.csv};
\addlegendentry{Filtered ERE}

\end{axis}
\end{tikzpicture}
\caption{
Success rate of decoding algorithms for parameters $q=2$, $n=30$ and $k=5$.
}
\label{fig:65_graph}
\end{figure}

From the two graphs, we can see how for \textit{ER} (the red line) the accuracy drops to zero after $12$ and $6$ insertions respectively as predicted by \eqref{eq:ER_max_t_duplicate}.
We can observe that the behavior near the threshold is governed by the effect of the floor operator: when $\frac{n-r}{k}$ is nearly integral the bound is loose and performance decays smoothly, while a strong truncation by the floor yields a much sharper transition from high accuracy to zero.
For \textit{ERE} (the blue line), according to \eqref{eq:ERE_max_t}, the decrease in accuracy is predicted to start from $13$ and $9$ insertions respectively. This matches very well in the first case (see Figure \ref{fig:88_graph}) while in the second case the decline starts already from $7$ (see Figure \ref{fig:65_graph}). Increasing the expansion size by one would keep the accuracy high until $10$ insertions in the second example, but a similarly larger expansion would give worse results in the first example.
In practice, it is hard to find a formula for the optimal expansion size in ERE that fits all combinations of parameters and dimension of deletion and insertion. This remains an open problem for future work.

For \textit{Filtered ERE} (the green line) the predicted decrease in performance should be for $27$ and $12$ insertions respectively.
From Figure \ref{fig:88_graph} and \ref{fig:65_graph} we can see the decrease starts a bit earlier in both cases.

Analyzing the reason of these failures, in most of the cases, we observe spaces of dimension $2k$ or more. This is because, for large insertion spaces, there is a non-negligible probability that their projection with respect to the canonical map $\pi_a\colon\Fqn\to\Fqn/a\Fqk$ is non-scattered with respect to the Desarguesian spread $\pi_a(\bm{\calD})=\{\pi_a(b\Fqk) \colon b\in\Fqn^*\text{ and } b\Fqk\cap a\Fqk=\{0\}\}$ in $\Fqn/a\Fqk$, as will be observed at the end of Section \ref{sec:Probability of success for a minimum distance decoder}.
As a consequence, there is more than just a single element of the spread having intersection at least $2$ with $\calR$.
This means the filter could sample with the same probability (or higher if the intersection is larger than $2$) from each of these spaces.
We can empirically affirm that the Filtered ERE works with high accuracy up to the theoretical upper bound where it is possible to perform unique decoding.
The small negative bump in accuracy we see in the first set of experiments between the $14$ and $16$ insertions can be improved by reducing the number of expansions.

\paragraph{Maximal iterations in ERE.}
The simple ERE performs well up to the bound in \eqref{eq:ERE_max_t} and does not drop to zero. Increasing the number of maximal iteration in ERE can greatly improve this algorithm even for insertions of size comparable to the one handled by Filtered ERE

In the following experiment we run ERE on the same instances of the problem but stopping after a different time of maximal iterations.

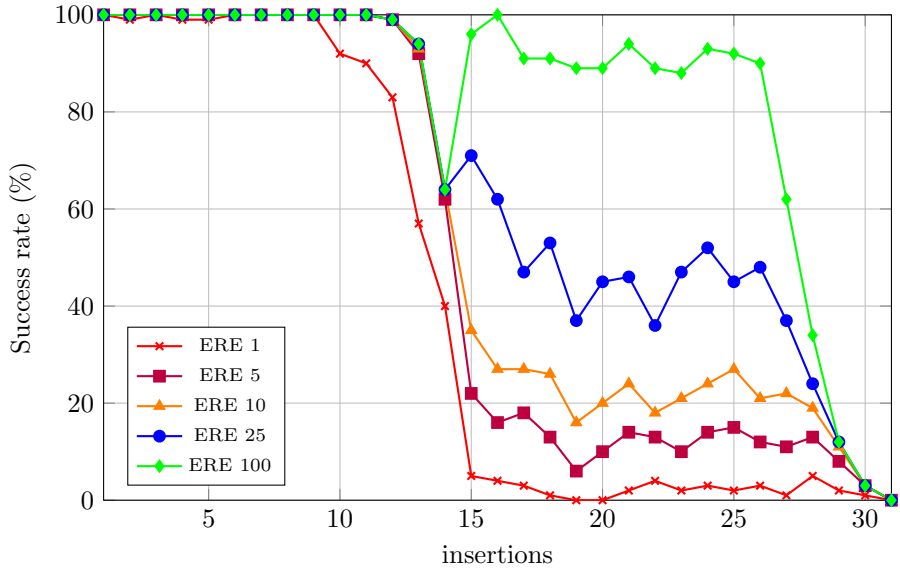
\begin{figure}[!ht]
\begin{tikzpicture}
\begin{axis}[
    width=12cm,
    height=8cm,
    xlabel={insertions},
    ylabel={Success rate (\%)},
    y filter/.code = \pgfmathparse{#1/1},
    xmin=1, xmax=31,
    ymin=0, ymax=100,
    grid=major,
    legend pos=south west,
    legend style={
        font=\scriptsize,
        row sep=1pt
    }
]

\addplot[
    thick,
    red,
    mark=x
] table [col sep=comma, x=insertions, y=ERE_1] {test_ERE.csv};
\addlegendentry{ERE 1}

\addplot[
    thick,
    purple,
    mark=square*
] table [col sep=comma, x=insertions, y=ERE_5] {test_ERE.csv};
\addlegendentry{ERE 5}

\addplot[
    thick,
    orange,
    mark=triangle*
] table [col sep=comma, x=insertions, y=ERE_10] {test_ERE.csv};
\addlegendentry{ERE 10}

\addplot[
    thick,
    blue,
    mark=*
] table [col sep=comma, x=insertions, y=ERE_25] {test_ERE.csv};
\addlegendentry{ERE 25}

\addplot[
    thick,
    green,
    mark=diamond*
] table [col sep=comma, x=insertions, y=ERE_100] {test_ERE.csv};
\addlegendentry{ERE 100}

\end{axis}
\end{tikzpicture}
\caption{
Success rate of the ERE algorithm for parameters $q=2$, $n=64$ and $k=8$ with different maximum numbers of iterations.
}
\label{fig:88_ERE_graph}
\end{figure}

For a large number of iteration ERE becomes almost as reliable as Filtered ERE even for large insertions.

\section{Success Probability for a Nearest Neighbour Decoder}
\label{sec:Probability of success for a minimum distance decoder}
The accuracy of our algorithms is dominated by the theoretical accuracy of a nearest neighbour decoder.
At the end of Section \ref{sec:A first decoding algorithm}, we have seen how the success of our algorithm strongly depends on the intersection behavior of $\calR$ with the elements of the Desarguesian $k$-spread in $\Fqn$. 
Let $\calR =\calU\oplus\calB$ where $\calU$ is an $\Fq$-subspace of $a \Fqk$ of dimension $k-d$ and $\calB$ is a subspace of dimension $t$ of $\Fqn$. We assume, without loss of generality, that $\calB\cap a\Fqk=\{0\}$, indeed, if this condition does not hold, then there exist a smaller subspace $\calB'\subsetneq \calB$ and a larger subspace $\calU'\supsetneq \calU$ such that
\[
\calR=\calB\oplus\calU=\calB'\oplus\calU'.
\]
Therefore, the actual dimensions of insertion and deletion is smaller than $\dim(\calB)$ and $k-\dim(\calU)$, respectively. Hence, we can always consider the direct sum as the worst-case scenario.
A nearest neighbour decoder will fail if there is a $b \in \Fqn^*$ such that $\ds(\calR, b\Fqk) \leq \ds(\calR, a\Fqk)$ and $b \Fqk \neq a \Fqk$.
Keeping in mind that $\ds(\calR, b\Fqk) = 2k - d + t - 2\dim(\calR \cap b\Fqk)$ and $\ds(\calR, a\Fqk) = d+t$, the failure condition can be expressed as
\begin{equation}
\label{eq:min_distance_decoder_failure}
    \dim(\calR \cap b\Fqk) \geq k - d.
\end{equation}
The optimal case for a nearest neighbour decoder and, in particular, for our algorithm to succeed, corresponds to the case where all the intersections of $\calR$ with the elements of $\bm \calD$ have dimensions strictly smaller than $k-d$, possibly as small as one, except for the intersection with $a\Fqk$. In other words, for our algorithm to succeed, we want that, for all 
$b \in \Fqn^*$ such that $b\Fqk \neq a\Fqk$, 
\begin{equation}
\label{eq:min_distance_decoder_maybe_success}
    \dim(\calR \cap b\Fqk) \leq \eta,
\end{equation}
for some parameter $\eta$ satisfying $1 \leq \eta < k-d$.
\\
\par Let us consider $a\in\Fqn^*$ such that $\calU\subseteq a\Fqk$ and let us consider the linear map, known as canonical quotient
\begin{equation*}
\pi_a\colon\Fqn \to\Fqn/a\Fqk,
\end{equation*}
where $\pi_a(b) = \pi_a(b')$ iff $b' - b \in a \Fqk$.  
Notice that the kernel of this map is exactly $a \Fqk$ as $ax - 0 \in a \Fqk$ for any $x \in \Fqk$.
In Proposition \ref{prop: evasivness still relies pi calB}, we will see how the condition expressed by Equation \eqref{eq:min_distance_decoder_failure} can be reformulated in terms of a well-known geometric property of the space $\pi_a(\calB)$. Before doing so, we recall some preliminary notions; in particular, we introduce the following definition, which first appeared in \cite{blokhuis2000scattered} and was subsequently generalized in \cite[Definition 2.2]{gruica2024generalised}.
\begin{definition}
    Let $1\leq \eta\leq k\leq n$ be positive integers. Let $\calG$ be a subset of $\mathrm{Gr}_{\Fq}(k,\F_{q^n})$ and let $\calV$ be an $\Fq$-subspace of $\F_{q^n}$. $\calV$ is $(\calG,\eta)$-evasive if
    \[
    \dim_q(\calV\cap S)\leq \eta\text{ for all }S\in\calG.
    \]
If $\calG$ is clear from the context, we will just say that the evasivness of $\calV$ is $\eta$.
\end{definition}
Having stated the general definition, we now provide its specialization to our setting.
\begin{definition}
\label{def: evasivness of V}
Let $1\leq\eta\leq k\leq n$ be positive integers. Let $\calV$ be an $\Fq$-linear subspace of $\Fqn$, where $n = kr$.
Considering the (Desarguesian) spread $\bm{\mathcal{D}}
= \{ b \Fqk \mid b \in \Fqn^* \} \subseteq \mathrm{Gr}_{\Fq}(k,\Fqn)$, we say that $\calV$ is $(\bm \calD,\eta)$-evasive if 
$$
    \dim_q(\calV \cap b \Fqk) \leq \eta \text{ for all }b \in \Fqn^*.
$$
If $\bm \calD$ is clear from the context, we will just say that the evasivness of $\calV$ is $\eta$. 
\end{definition}
Significant results on $(\mathcal{G},\eta)$-evasive spaces focus on bounds on their dimension. In particular, we recall \cite[Corollary 4.9]{bartoli2021evasive} which applies to the Desarguesian spread $\bm \calD $.
\begin{theorem}
    \label{ref: bound on the dimnesion of an evasive subspace with respect to a esarguesian spread}
    Let $n=kr$ and let $\bm \calD = \{ b \Fqk \mid b \in \Fqn^* \} \subseteq \mathrm{Gr}_{\Fq}(k,\Fqn)$. Let $\eta$ be an integer such that $1\leq \eta\leq k$ and $\calV$ be an $\Fq$-subspace of $\Fqn$. If $\calV$ is $(\bm \calD,\eta)$-evasive, then
    \[
    \dim_q(\calV)\leq \frac{\eta kr}{\eta+1}.
    \]
\end{theorem}
We will now see how the canonical quotient map $\pi_a$ induces a Desarguesian spread on $\Fqn/a\Fqk$. 
\begin{proposition}
\label{prop: pi is also a spread}
Let $n=kr$, $\bm \calD = \{ b \Fqk \mid b \in \Fqn^* \} \subseteq \mathrm{Gr}_{\Fq}(k,\Fqn)$ and let $\pi_a\colon\Fqn\to\Fqn/a\Fqk$ be the canonical quotient map. Then the family
\[
\pi_a(\bm\calD)
=
\{\pi_a(b\Fqk)\colon  b\in\Fqn^*\text{ and }b\Fqk\cap a\Fqk=\{0\}\}
\]
is a Desarguesian $k$-spread of the quotient space $\Fqn/a\Fqk$. 
\end{proposition}

\begin{proof}
First, let us consider $b\in\Fqn^*$ such that $b\Fqk\cap a\Fqk=\{0\}$. Since $\pi_a$ is an $\Fqk$-linear map and the space $b\Fqk$ has trivial intersection with the kernel of $\pi_a$, we have
\[
\dim_{\Fqk}(\pi_a(b\Fqk))=\dim_{\Fqk}(b\Fqk)=1.
\]
Now we prove that distinct elements of $\pi_a(\bm \calD)$ have
trivial intersection. 
Let $b,c\in\Fqn^*$ and assume that $\pi_a(b\Fqk)\cap
\pi_a(c\Fqk)\neq\{0\},$ then there exist nonzero elements $bx\in b\Fqk$ and $cy\in c\Fqk$, for some $x,y \in \Fqk$, such that $\pi_a(bx)=\pi_a(cy)$,  i.e.,  $bx-cy \in a\Fqk.$
In other words there exists $z \in \Fqk$ such that
$ bx - cy = az $. Multiplying on both sides by $x^{-1}t$ for $t \in \Fqk$, we obtain $\pi_a(bt) \in \pi(c \Fqk)$, which implies that $\pi_a(b\Fqk ) \subseteq \pi_a(c \Fqk)$. By symmetry, one also obtains that $\pi_a(b\Fqk) \supseteq \pi_a(c\Fqk)$. Therefore, $\pi_a(b\Fqk)$ and $\pi_a(c\Fqk)$ either coincide or intersect trivially.
To sum up, $\pi_a(\bm \calD)$ is a collection of
$1$-dimensional $\Fqk$-subspaces of $\Fqn/a\Fqk$ which partitions all nonzero elements, i.e., it is a Desarguesian $k$-spread of $\Fqn/a\Fqk$.
\end{proof}
 Let us consider the partial spread $$\bm\calD_a = \{b \Fqk \mid b \in \Fqn, b\Fqk \neq a\Fqk \}.$$ Thanks to the previous proposition, we can prove that the evasiveness of $\calB \oplus a\Fqk$ with respect to the partial spread $\bm\calD_a$ in $\Fqn$ is the same as the evasiveness of $\pi_a(\calB)$ with respect to the Desarguesian spread $\pi_a(\bm\calD)$ in $\Fqn/a\Fqk$. 
\begin{proposition}
\label{prop: evasivness still relies pi calB}
Let $n=kr$, $\bm \calD = \{ b \Fqk \mid b \in \Fqn^* \} \subseteq \mathrm{Gr}_{\Fq}(k,\Fqn)$ and let $\pi_a\colon\Fqn\to\Fqn/a\Fqk$ be the canonical quotient map. Let also $\pi_a(\bm \calD)
=
\{\pi_a(b\Fqk)\colon  b\in\Fqn^*\text{ and }b\Fqk\cap a\Fqk=\{0\}\}$ be the $k$-spread of the quotient space $\Fqn/a\Fqk$ from Proposition \ref{prop: pi is also a spread}. Let $\calB$ be an $\Fq$-subspace of $\Fqn$ such that $\calB\cap a\Fqk=\{0\}.$ Then, for every
$b\in\Fqn^*$ such that $b\Fqk\cap a\Fqk=\{0\}$, we have that
\[
\dim\left(\pi_a(\calB)\cap\pi_a(b\Fqk)\right)
=
\dim\left((\mathcal{B}\oplus a\Fqk)\cap b\Fqk\right).
\]
Consequently the evasiveness of $(\calB \oplus a\Fqk)$ with respect to the partial spread $\bm\calD_a = \{b \Fqk \mid b \in \Fqn, b\Fqk \neq a\Fqk \}$ is the same as the evasiveness of $\pi_a(\calB)$ with respect to the Desarguesian spread $\pi_a(\bm\calD)$ in $\Fqn/a\Fqk$.
\end{proposition}
\begin{proof}
By the properties of the canonical quotient map, 
\[
\pi_a(\calB)\cap\pi_a(
b\Fqk)
=
\pi_a(
(\calB +a\Fqk)
\cap
(b\Fqk + a\Fqk)
)
\]
where, since $\calB\cap a\Fqk=\{0\}$ and $b\Fqk\cap a\Fqk=\{0\}$, then
\[
(\calB\oplus a\Fqk)
\cap
(b\Fqk\oplus a\Fqk)
=
\left((\calB\oplus a\Fqk)
\cap
b\Fqk\right)
\oplus
a\Fqk.
\]
Taking the quotient modulo $a\Fqk$ gives the following
isomorphism
\[
\pi_a(
(\calB\oplus a\Fqk)
\cap
(b\Fqk\oplus a\Fqk)
)
\simeq 
(\calB\oplus a\Fqk)
\cap
b\Fqk.
\]
Hence,
\[
\dim\left(\pi_a(\calB)\cap\pi_a(b\Fqk)\right)
=
\dim\left((\mathcal{B}\oplus a\Fqk)\cap b\Fqk\right).
\]
\end{proof}
\begin{corollary}
    \label{cor: consequences on evasivness of the received space}
    Let $n=kr$, $\bm \calD = \{ b \Fqk \mid b \in \Fqn^* \} \subseteq \mathrm{Gr}_{\Fq}(k,\Fqn)$ and let $\pi_a\colon\Fqn\to\Fqn/a\Fqk$ be the canonical quotient map. Let also $\pi_a(\bm\calD)
=
\{\pi_a(b\Fqk)\colon  b\in\Fqn^*\text{ and }b\Fqk\cap a\Fqk=\{0\}\}$ be the $k$-spread of the quotient space $\Fqn/a\Fqk$ from Proposition \ref{prop: pi is also a spread}. Let $\calR =\calU\oplus\calB$ where $\calU$ is an $\Fq$ subspace of $a \Fqk$ of dimension $k-d$ and $\calB$ is a subspace of dimension $t$ not intersecting $a \Fqk$. Then for all 
$b \in \Fqn^*$ such that $b\Fqk \neq a\Fqk$, 
\begin{equation}
\label{eq:min_distance_decoder_maybe_success}
    \dim(\calR \cap b\Fqk) \leq \eta,
\end{equation}
where $\eta$ is the evasiveness of $\pi_a(\calB)$ with respect to the spread $\pi_a(\calD)$. 
\end{corollary}
To conclude this section, we present the following lemma, which will be crucial in the remainder of the paper. 
\begin{lemma}
\label{lemma: B random also piB random}
Let $n=kr$, $a\in\Fqn^*$ and let $\pi_a\colon\Fqn\to\Fqn/a\Fqk$ be the canonical quotient map. Let $\calB$ be an $\Fq$-subspace of dimension $t$ chosen uniformly at random in $\Fqn$ such that $\calB\cap a\Fqk=\{0\}.$
Then $\pi_a(\mathcal{B})$ is uniformly distributed among the
$t$-dimensional $\F_q$-subspaces of the quotient space $\Fqn/a\Fqk$.
\end{lemma}

\begin{proof}
We extend an ordered basis $\bm v_a$ of $a\Fqk$ to an ordered basis of $\Fqn$, which we denote by $\bm v=(\bm v_0,\bm v_a).$
The basis $\bm v$ naturally defines an isomorphism $\phi_{\bm v}:\Fqn\to\Fq^n,$ which sends every element of $\Fqn$ to its coordinate representation with respect to $\bm v$. Therefore, each $t$-dimensional subspace $\mathcal{B}$ can be represented by a generator matrix $G \in \mathbb{F}_q^{t \times n}$, whose rows are the images under $\phi_v$ of the elements of a basis of $\mathcal{B}$. This representation becomes unique by considering the reduced row echelon form $\widetilde{G}$ of $G$.
As $\calB \cap a \Fqk = \{0\}$, then
$$
\widetilde{G} = 
\left (
    \begin{array}{c|c}
    E & G_a
    \end{array}
\right )
$$
where $E\in \Fq^{t \times {(n-k)}}$ is a reduced row echelon form of rank $t$ and $G_a \in \Fq^{t \times k}$. Observe that if $\rank(E) < t$, then the last row of $E$ would be zero, which would imply that the intersection $\calB \cap a \Fqk$ is nontrivial. If we replace the matrix $G_a \in \Fq^{t \times k}$ with any other matrix of the same size, we obtain exactly $q^{tk}$ distinct matrices of the same form. These matrices correspond to distinct subspaces $\hat{\calB}$ satisfying the same properties as $\calB$ and such that $\pi_a(\calB) = \pi_a(\hat{\calB})$. Since every $t$-dimensional subspace of $\Fqn/a\Fqk$ has exactly $q^{tk}$ pre-images in $\{\calB\subseteq \mathbb F_{q^n}\colon\dim(\calB)=t,\calB\cap a\Fqk=\{0\}\}$, each quotient subspace is obtained with the same probability. Therefore, the induced distribution of $\pi_a(\calB)$ is uniform.
\end{proof}

\subsection{Probability that a Random Subspace has Evasiveness $\eta$}
    Following the discussion in the previous subsection, a nearest neighbour decoder will give the correct answer only if for all 
$b \in \Fqn^*$ such that $b\Fqk \neq a\Fqk$, 
\[
\dim(\calR \cap b\Fqk) \leq \eta,
\]
for some parameter $\eta$ satisfying $1 \leq \eta < k-d$. 
On the other hand, Corollary \ref{cor: consequences on evasivness of the received space} provides an upper bound on $\dim(\calR \cap b\Fqk)$ for all $b \in \Fqn^*$ such that $b\Fqk \neq a\Fqk$, given by the evasiveness of $\pi_a(\calB)$ with respect to the Desarguesian spread $\pi_a(\calD)$.
Additionally, Lemma \ref{lemma: B random also piB random} states that, if $\calB$ is a $t$-dimensional $\Fq$-subspace of $\Fqn$ chosen uniformly at random in $\Fqn$, then also $\pi_a(\calB)$ is chosen uniformly at random among the
$t$-dimensional $\F_q$-subspaces of the quotient space $\Fqn/a\Fqk$. 
To estimate the probability of success of a nearest neighbour decoder we should then count the number of spaces $\pi_a(\calB) \in \Gr_{\Fq}(i,\Fqn/a\Fqk)$ with evasiveness $\eta$, for some $1\leq\eta < k-d.$
\\
\par 
For ease of notation, we will count the number of subspaces $\calB \in \Gr_{\Fq}(i,\Fqn)$ with evasiveness $\eta$ with respect to the Desarguesian spread $\bm \calD$, for some $1 \leq \eta < k-d$ and, at the end of this section, we will translate the obtained results back to the original setting. 
\\
\par Consider the set
    $$
        S_{i,\eta} := \{\calB \in  \Gr_{\Fq}(i,\Fqn) \mid \dim(\calB \cap a \Fqk) \leq \eta, \forall a \in \Fqn^*\}.
    $$
The probability that a random $\calB$ of dimension $i$ has evasiveness upper bounded by $\eta$ will be given by
    $$
        \frac{|S_{i,\eta}|}{\Gr_{\Fq}(i,\Fqn)} = 
        \frac{|S_{i,\eta}|}{\gbinom{n}{i}_q},
    $$
    where 
    $$
        \gbinom{n}{i}_q = \prod_{j = 0}^{i-1} \frac{q^n - q^j}{q^i - q^j},
    $$
is the Gaussian coefficient, which is known to count the number of subspaces of $\Fq$ dimension $i$ in $\Fqn$. Thanks to Theorem \ref{ref: bound on the dimnesion of an evasive subspace with respect to a esarguesian spread}, we know that $|S_{i,\eta}|=0$ for all $i\in\left\{\left\lfloor\frac{\eta kr}{\eta+1}\right\rfloor+1,\ldots,n
\right\}$. To determine $|S_{i,\eta}|$ for $i\leq\left\lfloor\frac{\eta kr}{\eta+1}\right\rfloor$, we will use the following lemma.
\begin{lemma}
        \label{lemma: counting intersections}
        Let $d_1,d_2,h,n$ be positive integers such that $0\leq h\leq \min\{d_1,d_2\}$ and $d_2-h\leq n-d_1$. Fix a $d_1$-dimensional subspace $\calV_1\leq\Fqn$. Then 
        \[
        \left\lvert\{d_2\text{-dimensional subspaces }\calV_2\leq\Fqn\text{ with }\mathrm{dim}(\calV_1\cap\calV_2)=h\}\right\rvert=q^{(d_1-h)(d_2-h)}\gbinom{n-d_1}{d_2-h}_q\gbinom{d_1}{h}_q.
        \]
    \end{lemma} 
    \begin{proof}
        We show this up to isomorphism. Let $e_i$ be the vector in $\Fq^n$ whose $i$-th coordinate is one and all other coordinates are zero. Denote by $\calE$ the $d_1$-dimensional subspace of $\Fq^n$ generated by $e_{n-d_1+1},\ldots,e_n$. Thanks to \cite[Lemma 2.1]{wang2010association}, we have that
        \[
        \left\lvert\left\{\calP \mid \dim(\calP)=d_2\text{ and }\dim(\calP\cap\calE)=h\right\}\right\rvert=q^{(d_1-h)(d_2-h)}\gbinom{n-d_1}{d_2-h}_q\gbinom{d_1}{h}_q.
        \]
        Since two vector spaces are isomorphic if and only if they have the same dimension, we let $\varphi$ be the isomorphism that maps $\calE$ into $\calV_1$, i.e., $\varphi(\calE)=\calV_1$ and extend it to an isomorphism of the whole space $\Fq^n$. The claim follows from the fact that isomorphisms preserve the dimensions of subspaces and that, for any subspace $\calP$ of dimension $d_2$,
        \[
            h = \dim(\calP \cap \calE) = \dim\big(\varphi(\calP\cap \calE)\big) = \dim\big(\varphi(\calP) \cap \calV_1\big),
        \]
        and conversely, for any $d_2$-dimensional subspace  $\calV_2$ of $\mathbb{F}_q^n$,
        \[
            h = \dim(\calV_1 \cap \calV_2) = \dim\big(\varphi^{-1}(\calV_1 \cap \calV_2)\big) = \dim\big(\varphi^{-1}(\calV_1) \cap \calE\big).
        \]
    \end{proof}
    Thanks to the latter we can provide the following proposition. \begin{proposition}
        \label{ref: our lower bound}
        Let $0\leq i\leq\left\lfloor\frac{\eta kr}{\eta+1}\right\rfloor$, then
        \[
            \left\lvert S_{i,\eta}\right\rvert\geq \gbinom{n}{i}_q-\frac{q^n-1}{q^k-1}\left(\gbinom{n}{i}_q-\sum_{h=0}^\eta q^{(k-h)(i-h)}\gbinom{n-k}{i-h}_q\gbinom{k}{h}_q\right)\eqqcolon \mathrm{LB}_{\eta}.
        \]
        This bound is tight if and only if $i \leq 2 \eta + 1$.
    \end{proposition}
    \begin{proof}
        Let $a\in\Fqn^*$, $h\in\{0,\ldots,\eta\}$ and define 
    \[
    \Sigma_{a,h}\coloneqq \{\calB \in  \Gr_{\Fq}(i,\Fqn) \mid \dim(\calB \cap a \Fqk)=h\}, 
    \]
    then
    \[
    S_{i,\eta}=\bigcap_{a\in\Fqn^*}\left(\bigcup_{h=0}^\eta\Sigma_{a,h}\right).
    \]
    Fix $\calB \in  \Gr_{\Fq}(i,\Fqn)$, the latter follows from the fact that for all $a \in \Fqn^*,~\dim(\calB \cap a \Fqk) \leq \eta$ if and only if for all $a \in \Fqn^*$, there exists $h\in\{0,\ldots,\eta\}$ such that $\dim(\calB \cap a \Fqk)=h$. We first note that distinct values of $a$ can correspond to the same element of the spread. Thus, we define $\mathfrak{A}\coloneqq\Fqn^*/\sim_{q^k}$ where $\sim_{q^k}$ denotes the following equivalence relation on $\Fqn^*$: 
    \[
    a\sim_{q^k} b\text{ if and only if }a\Fqk=b\Fqk.
    \]
    Therefore, $S_{i,\eta}$ can be also seen as
    \[
    S_{i,\eta}=\bigcap_{a\in\mathfrak{A}}\left(\bigcup_{h=0}^\eta\Sigma_{a,h}\right).
    \]
    Recall we want to determine $|S_{i,\eta}|$. Once fixed $a\in\mathfrak{A}$, since for all $h,j\in\{0,\ldots,\eta\}$ with $h\neq j$, $\Sigma_{a,h}\cap\Sigma_{a,j}=\emptyset$, and thanks to Lemma \ref{lemma: counting intersections}, we have that
    \[
    \left\lvert\bigcup_{h=0}^\eta\Sigma_{a,h}\right\rvert=\sum_{h=0}^\eta\left\lvert\Sigma_{a,h}\right\rvert=\sum_{h=0}^\eta q^{(k-h)(i-h)}\gbinom{n-k}{i-h}_q\gbinom{k}{h}_q.
    \]
    As dealing with the intersections for $a$ varying in $\mathfrak{A}$ is nontrivial, we consider the complement with respect to the set $\Gr_{\Fq}(i,\Fqn)$
    \[
    S_{i,\eta}^c=\bigcup_{a\in\mathfrak{A}}\left(\bigcup_{h=0}^\eta\Sigma_{a,h}\right)^c.
    \]
    Therefore 
    \begin{equation}
    \label{eq: evasive space inequality bound}
        \begin{aligned}
            \left\lvert S_{i,\eta}^c\right\rvert&=\left\lvert \bigcup_{a\in\mathfrak{A}}\left(\bigcup_{h=0}^\eta\Sigma_{a,h}\right)^c\right\rvert\leq \sum_{a\in\mathfrak{A}} \left\lvert \left(\bigcup_{h=0}^\eta\Sigma_{a,h}\right)^c\right\rvert\\
            &=\frac{q^n-1}{q^k-1}\left(\gbinom{n}{i}_q-\sum_{h=0}^\eta q^{(k-h)(i-h)}\gbinom{n-k}{i-h}_q\gbinom{k}{h}_q\right),
        \end{aligned}
    \end{equation}
    where $\frac{q^n - 1}{q^k - 1} = |\mathfrak{A}|.$
    From the relation $|S_{i,\eta}| = \gbinom{n}{i}_q -  |S_{i,\eta}^c| $ we obtain the desired result.
    \par In (\ref{eq: evasive space inequality bound}) equality is satisfied if and only if all sets $\left(\bigcup_{h=0}^\eta\Sigma_{a,h}\right)^c$ are disjoint for $a \in \mathfrak{A}$.
    A necessary and sufficient condition for this to be true is given by $i \leq 2 \eta + 1$.
    We can describe the set $\left(\bigcup_{h=0}^\eta\Sigma_{a,h}\right)^c$ as the set of all subspaces of dimension $i$ that intersect $a \Fqk$ with dimension at least $\eta + 1$.
    For $i = 2 \eta + 2$ we can easily choose two subspaces of dimension $\eta + 1$ from two distinct spaces $a \Fqk \neq b \Fqk,$ their sum will have dimension $2 \eta + 2$ and will lie in the intersection of $\left(\bigcup_{h=0}^\eta\Sigma_{a,h}\right)^c \cap \left(\bigcup_{h=0}^\eta\Sigma_{b,h}\right)^c$ showing the condition is necessary.
    
    To show it is also sufficient, consider $\calB \in \left(\bigcup_{h=0}^\eta\Sigma_{a,h}\right)^c$ of dimension $\dim(\calB) = i \leq 2 \eta + 1$ and a subspace of the form $b \Fqk \neq a \Fqk$.
    Since $a \Fqk \cap b \Fqk = \{0\}$ then $\calB \cap a \Fqk$ and $\calB \cap b \Fqk$ will have zero intersection, while their sum will have, at most, dimension $i \leq 2\eta + 1$.
    From which
    \[
      \dim(\calB \cap b \Fqk) \leq  2 \eta + 1  - \dim(\calB \cap  a\Fqk) \leq \eta, 
    \]
    where for the last inequality we used $\dim(\calB \cap a \Fqk) \geq \eta + 1.$
    This ensures that $\calB \notin \left(\bigcup_{h=0}^\eta\Sigma_{b,h}\right)^c$ for any $b \neq a \in \mathfrak{A}$, implying that sets $\left(\bigcup_{h=0}^\eta\Sigma_{b,h}\right)^c$ are always disjoint for $i \leq 2 \eta +1$.
    \end{proof}
   
    In \cite{gruica2024generalised} the authors derive bounds on the number of $(\calG,\eta)$-evasive spaces for a partial $k$-spread $\calG$ in $\Fqn$ in a different way. These bounds rely on \cite[Lemma 5.11, Lemma 5.12]{gruica2024generalised}), which we summarize in the following lemma with the parameters relevant to this work.
\begin{lemma}{\cite{gruica2024generalised}}
    Let $r\geq 2$, $1 \leq i \leq n$ and $1 \leq k \leq kr-k=k(r-1)$ be integers and let $S$ be a $k$-dimensional subspace in $\Fqn$.
    \noindent The number of $i$-spaces in $\Fqn$ that intersect $S$ in dimension at least $\eta+1$ is given by
    \[
    \partial_q(n,i,k,\eta) \coloneqq \sum_{\ell=\eta+1}^{k} \sum_{b=\ell}^{k} \gbinom{k}{\ell}_q \gbinom{k-\ell}{b-\ell}_q \gbinom{n-b}{i-b}_q (-1)^{b-\ell} q^{\binom{b-\ell}{2}}.
\]
Moreover, let $S,S^\prime$ be $k$-dimensional subspaces in $\Fqn$ with $S\cap S^\prime=\{0\}$. For an integer $1\leq\tilde{\eta}\leq k-1$ the number of $i$-spaces in $\Fqn$ that intersect both $S$ and $S'$ in dimension at least $\tilde{\eta}+1$ is $\omega_q(n,i,k,\tilde{\eta})$, which is given by
\[
\sum_{\ell=\tilde{\eta}+1}^{k} \sum_{\ell^{\prime}=\tilde{\eta}+1}^{k} \gbinom{k}{\ell}_q \gbinom{k}{\ell^{\prime}}_q \sum_{r=\ell}^{k} \sum_{s=\ell^{\prime}}^{k} \gbinom{k-\ell}{r-\ell}_q \gbinom{k-\ell^{\prime}}{s-\ell^{\prime}}_q
\gbinom{n-r-s}{i-r-s}_q (-1)^{r+s-\ell-\ell^{\prime}} q^{\binom{r-\ell}{2} + \binom{s-\ell^{\prime}}{2}}
\]
\end{lemma}
As a consequence of the previous lemma and of \cite[Lemma 5.8]{gruica2024generalised} and \cite[Lemma 5.9]{gruica2024generalised} (resp.) they provide \cite[Corollary 5.14]{gruica2024generalised} and \cite[Corollary 5.18]{gruica2024generalised} (resp.) which we state compactly in the specific case of the Desarguesian spread under consideration.
\begin{corollary}{\cite{gruica2024generalised}}
    Let $\bm \calD= \{ a \Fqk \mid a \in \Fqn^* \}$ and let $1 \leq \eta \leq k-1$ be an integer. The number $\left\lvert S_{i,\eta}\right\rvert$ of $(\bm \calD, \eta)$-evasive $i$-spaces in $\Fqn$ is at least
    \begin{align*}
        \mathrm{LB}_{_{\eta},\mathrm{Gev}}&\coloneqq\gbinom{n}{i}_q - \frac{q^n-1}{q^k-1} \partial_q(n,i,k,\eta)\\
        &=\gbinom{n}{i}_q - \frac{q^n-1}{q^k-1}\left(\sum_{\ell=\eta+1}^{k} \sum_{b=\ell}^{k} \gbinom{k}{\ell}_q \gbinom{k-\ell}{b-\ell}_q \gbinom{n-b}{i-b}_q (-1)^{b-\ell} q^{\binom{b-\ell}{2}}\right)
    \end{align*}
    and at most
\[
\mathrm{UB}_{_{\eta},\mathrm{Gev}}\coloneqq\gbinom{n}{i}_q - \frac{\left(\frac{q^n-1}{q^k-1}\right) \partial_q(n,i,k,\eta)^2}{\partial_q(n,i,k,\eta) + \left(\left(\frac{q^n-1}{q^k-1}\right)-1\right) \omega_q(n,i,k,\eta)}.
\]
\end{corollary}
It turns out that, in the case of the Desarguesian spread $\bm \calD= \{ a \Fqk \mid a \in \Fqn^* \}$, we have $\mathrm{LB}_{_{\eta},\mathrm{Gev}} = \mathrm{LB}_\eta$ due to the following proposition. 
\begin{proposition}
    Let $r\geq 2$, $1 \leq i \leq n$ and $1 \leq k \leq kr-k=k(r-1)$ be integers and let $a\Fqk$ be an element of the Desarguesian spread $\bm \calD$.
    \noindent The number $\delta_q(n,i,k,\eta)$ of $i$-spaces in $\Fqn$ that intersect $a\Fqk$ in dimension at least $\eta+1$ is
    \[
    \gbinom{n}{i}_q-\sum_{h=0}^\eta q^{(k-h)(i-h)}\gbinom{n-k}{i-h}_q\gbinom{k}{h}_q.
\]
\end{proposition}
\begin{proof}
    Recalling the notation introduced in the proof of Proposition \ref{ref: our lower bound} we immediately have that the number of subspaces of dimension $i$ in $\Fqn$ that intersect an element $a\Fqk$ of the Desarguesian spread $\bm \calD$ in dimension at least $\eta+1$ is
    \begin{align*}
    \delta_q(n,i,k,\eta)&=\left\lvert\{\calB \in  \Gr_{\Fq}(i,\Fqn) \colon \dim(\calB \cap a \Fqk)\geq \eta+1\}\right\rvert\\
    &=\left\lvert\{\calB \in  \Gr_{\Fq}(i,\Fqn) \colon \dim(\calB \cap a \Fqk)=h\text{ where }h\in\{\eta+1,\ldots,k\}\}\right\rvert\\
    &=\left\lvert \left(\bigcup_{h=0}^\eta\Sigma_{a,h}\right)^c\right\rvert=\gbinom{n}{i}_q-\sum_{h=0}^\eta\left\lvert\Sigma_{a,h}\right\rvert=\gbinom{n}{i}_q-\sum_{h=0}^\eta q^{(k-h)(i-h)}\gbinom{n-k}{i-h}_q\gbinom{k}{h}_q.
\end{align*}
\end{proof}
In summary, the bounds $\mathrm{LB}_{_{\eta}}$, $\mathrm{LB}_{_{\eta},\mathrm{Gev}}$, and $\mathrm{UB}_{_{\eta},\mathrm{Gev}}$ satisfy
\begin{equation}
    \label{ref: evasive bounds}
    0<\mathrm{LB}_{_{\eta},\mathrm{Gev}}=\mathrm{LB}_{_{\eta}}\leq \mathrm{UB}_{_{\eta},\mathrm{Gev}}\leq \gbinom{n}{i}_q. 
\end{equation}
Moreover, when $i \leq 2 \eta + 1$, 
\[
0<\mathrm{LB}_{_{\eta},\mathrm{Gev}}=\mathrm{LB}_{_{\eta}}= \mathrm{UB}_{_{\eta},\mathrm{Gev}}\leq \gbinom{n}{i}_q, 
\]
i.e., we have the exact number of $(\bm \calD, \eta)$-evasive subspaces.

As a consequence of Equation \ref{ref: evasive bounds}, the probability that a randomly chosen $\calB \in \Gr_{\Fq}(i,\Fqn)$ has evasiveness $1\leq \eta\leq k-1$ is bounded above and below by
\begin{equation}
    \label{eq: prob of being (D,eta) evasive}
    \frac{\mathrm{LB}_{_{\eta}}}{\gbinom{n}{i}_q}\leq  \mathrm{P}(\calB \text{ is } (\bm \calD,\eta)\text{-evasive})
\leq \frac{\mathrm{UB}_{_{\eta},\mathrm{Gev}}}{\gbinom{n}{i}_q}.
\end{equation}
\subsubsection{Final Remarks on the Success Probability}
According to the considerations at the beginning of this section and as a consequence of Corollary \ref{cor: consequences on evasivness of the received space}, $\mathrm{P}(\pi_a(\calB) \text{ is } (\pi_a(\bm \calD),\eta)\text{-evasive})$ for $1 \leq \eta < k-d$ corresponds to the probability that a nearest neighbour decoder outputs a unique codeword and that this codeword is the transmitted one.

\medskip
Let $q$ be a prime power. Given a sequence $(\calA_q)_{q}$ of partial $k$-spreads with $\lvert\calA_q\rvert\geq 2$ for all $q$, the authors of \cite{gruica2024generalised} investigated how the behavior of the proportion of $(\calA_q,\eta)$-evasive subspaces within $\Gr_{\Fq}(i,\Fqn)$ relies on the asymptotic of the sequence $\left(\lvert\calA_q\rvert\right)_q$ as $q\rightarrow\infty$. 
In the particular case of the spread $\pi_a(\bm \calD_q)= \{\pi_a(b\Fqk)\colon  b\in\Fqn^*\text{ and }b\Fqk\cap a\Fqk=\{0\}\}\subseteq\Gr_{\Fq}(k,\Fqn/a\Fqk)\simeq\Gr_{\Fq}(k,\F_{q^{n-k}})$ for all $q$ prime power, then
\begin{equation}
    \label{eq: threshold dimension}
\lim_{q\rightarrow\infty}\frac{\left\lvert\left\{\calB_q\subseteq\Fqn/a\Fqk\colon\dim(\calB_q)=i, \calB_q\text{ is }(\pi_a(\bm \calD_q),\eta)\text{-evasive}\right\}\right\rvert}{\gbinom{n-k}{i}_q}=\begin{cases}
    1\text{ if }i\leq \gamma,\\
    0\text{ if }i\geq \gamma+2,
\end{cases}
\end{equation}
where $\gamma=\frac{\eta}{\eta+1}(n-2k)+\eta$. There is a threshold dimension at which $(\pi_a(\bm \calD_q),\eta)$-evasive subspaces transition from being dense to sparse. Consequently recalling Corollary \ref{cor: consequences on evasivness of the received space}, a nearest neighbour decoder will succeed if the projection of the insertion $\calB$ has evasiveness $\eta$ with respect to $\pi_a(\bm \calD_q)$ for some $1\leq\eta < k-d$. Consequently the latter provides us the threshold $\gamma$ on the dimension of the insertion below which the success of a nearest neighbour decoder, and hence of our algorithm, is guaranteed. 
\afterpage{
\begin{figure}[htbp]
    \begin{center}
    \begin{tikzpicture}[
        scale=0.85, transform shape,
        block/.style={
            draw=black, thick, rectangle,
            text width=5.5cm, minimum height=1.2cm,
            rounded corners=5pt, align=center,
            fill=white, font=\small
        },
        lightgrayblock/.style={
            block, fill=gray!12,
            text width=6.2cm, minimum height=2.0cm
        },
        arrow/.style={
            ->, >=Stealth, thick
        },
        doublearrow/.style={
            <->, >=Stealth, thick
        },
        container/.style={
            draw=black,
            rectangle,
            rounded corners=12pt,
            thick,
            inner sep=18pt
        }
    ]

    \node[lightgrayblock] (b1) at (0,6.8)
    {
    NND fails if there exists $b \in \mathbb{F}_{q^n}^*$ such that 
    $b\mathbb{F}_{q^k} \neq a\mathbb{F}_{q^k}$ and 
    $\dim(\mathcal{R} \cap b\mathbb{F}_{q^k}) \geq k - d$
    };

    \node[lightgrayblock] (b2) at (0,3.9)
    {
    A necessary condition to have a successful NND is that for all 
    $b \in \mathbb{F}_{q^n}^*$ such that 
    $b\mathbb{F}_{q^k} \neq a\mathbb{F}_{q^k}$, 
    $\dim(\mathcal{R} \cap b\mathbb{F}_{q^k}) \leq \eta$ 
    with $1 \leq \eta < k - d$
    };

    \draw[arrow] (b1) -- (b2);

    \node[lightgrayblock] (b3) at (0,0.8)
    {
    The projection $\pi_a(\calB)$ of the insertion space $\mathcal{B}$ is 
    $(\pi_a(\bm \calD),\eta)$-evasive for $1 \leq \eta < k - d$
    };

    \draw[doublearrow] (b2) -- 
    node[midway, right, font=\small] 
    {Corollary \ref{cor: consequences on evasivness of the received space}} 
    (b3);

    \node[lightgrayblock] (b4) at (0,-2.1)
    {
    Asymptotic results on the density of 
    $(\pi_a(\bm \calD),\eta)$-evasive subspaces: there is a threshold dimension 
    at which $\eta$-evasive subspaces with respect to the Desarguesian spread 
    $\pi_a(\bm \calD)$ in $\Fqn/a\Fqk$ transition from being dense to sparse.
    };

    \draw[arrow] (b3) -- (b4);

    \node[
        container,
        fit={(b1)(b2)(b3)(b4)},
        inner ysep=20pt,
        inner xsep=20pt
    ] (box) {};

    \end{tikzpicture}
    \end{center}
    \caption{Success probability for a nearest neighbour decoder}
    \label{fig: Probability of success for a nearest neighbour decoder}
\end{figure}
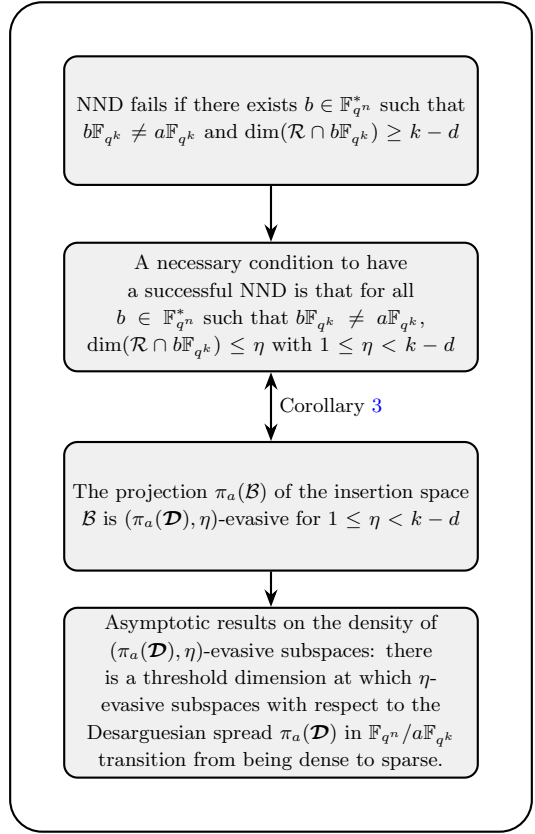}

\section{Conclusions and Open Problems}
\label{sec: Conclusions and Open Problems}
The main contribution of this work is a probabilistic polynomial-time decoding algorithm for Desarguesian spread codes, along with two refined version. Under a random insertion model, see Table \ref{tab:insertions-handled}, the proposed decoders are able to correct errors beyond half the minimum distance and, unlike previous approaches, can also handle received subspaces of dimension larger than $k$, consequently addressing the open problem in \cite{Gorla11}. The success probability and complexity of the proposed decoders are also discussed.

Future studies may extend the proposed algorithm (and relative refinements) to design efficient decoding algorithms for broader classes of cyclic subspace codes beyond Desarguesian spread codes, such as those constructed in \cite{roth2017construction}.   
\section*{Acknowledgments}
We would like to thank Hugo Beeloo-Sauerbier Couvée and Violetta Weger for fruitful discussions and suggestions.

\bibliographystyle{abbrv}
\bibliography{Bibliography.bib}

\appendix
\section{Appendix}
\addcontentsline{toc}{section}{Appendix}
\subsection{Intuitive Overview of the Relationship between the ER and the NND Algorithms}
\label{app: visual expl}
Let us (graphically) examine the reason why, for any value of $b \in \Fqn^*$ such that $b\Fqk\cap a\Fqk=\{0\}$, it is desirable for the dimension $\dim(\calR \cap b \Fqk)$ to be as small as possible. Suppose that the received subspace is $$\mathcal{R}=\calU\oplus\calB=\Span{\Fq}{au_1, \ldots, au_{k-d}} \oplus \Span{\Fq}{b_1, \ldots, b_{t}},$$
where $a \in\Fqk^* $, $u_i\in\Fqk$ are $\Fq$-linearly independent and $b_1, \ldots, b_{t}$ are $t$ $\Fq$-linearly independent elements chosen uniformly at random in $\mathbb{F}_{q^n}$. Since a Desarguesian spread $\bm  \calD$ form a partition of $\mathbb{F}_{q^n}^*$, we have that
\[
\mathcal{R}=\mathcal{R}\cap\mathbb{F}_{q^n}=\mathcal{R}\cap\left(\bigcup_{b\in\mathbb{F}_{q^n}^*}b\mathbb{F}_{q^k}\right)=\bigcup_{b\in\mathbb{F}_{q^n}^*}(\mathcal{R}\cap b\mathbb{F}_{q^k}),
\]
where each subset is disjoint if we remove the zero element.
In other words we can divide the space $\mathcal{R}$ in its intersections with each codeword of  $\bm  \calD$.  For the purpose of understanding, for each $b\in\Fqn^*$ such that $b\Fqk\cap a\Fqk=\{0\}$, we represent $\calR \cap b \Fqk$ by a red rectangle whose height is given by $\dim(\calR \cap b \Fqk)$, while $\calU \subseteq a \Fqk$ will be represented by a single green rectangle.
    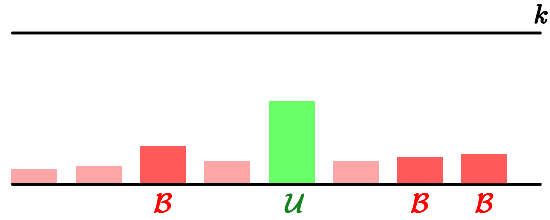
\begin{figure}[H]
        \centering
        \begin{tikzpicture}[x=1cm,y=1cm]

        \pgfmathsetmacro{\H}{2}        
        \pgfmathsetmacro{\barw}{0.6}   
        \pgfmathsetmacro{\gap}{0.25}   

        \def\levels{0.1, 0.12,0.25,0.15,0.55,0.15,0.18, 0.20}

        \newcount\n
        \n=0
        \foreach \h in \levels { \advance\n by 1 }
        \pgfmathsetmacro{\nbars}{\the\n}
        \pgfmathsetmacro{\W}{\nbars*(\barw+\gap)-\gap}

        \foreach[count=\i] \h in \levels {
            \pgfmathsetmacro{\xL}{(\i-1)*(\barw+\gap)}   
            \pgfmathsetmacro{\xR}{\xL+\barw}             
            \pgfmathsetmacro{\yT}{\h*\H}                 

            \def\thiscolor{red!35}
             
            \ifnum\i=3
                \def\thiscolor{red!65}
            \fi

            \ifnum\i=7
                \def\thiscolor{red!65}
            \fi

            \ifnum\i=8
                \def\thiscolor{red!65}
            \fi
            
            \ifnum\i=5
                \def\thiscolor{green!60}
            \fi

            \node[below] at (3.75,0) {$\color{darkgreen} \calU$};
            \node[below] at (2,0) {$\color{red} \calB$};
            \node[below] at (5.40,0) {$\color{red} \calB$};
            \node[below] at (6.25,0) {$\color{red} \calB$};

            \fill[\thiscolor] (\xL,0) rectangle (\xR,\yT);

            \draw[thick] (0,0) -- (7,0);                 
            \draw[thick] (0,\H) -- (7,\H) node[above]{$k$}; 
        }
    \end{tikzpicture}
        \caption{Received space $\calR=\calU \oplus \calB$}
        \label{fig: Received space R=U+B}
    \end{figure}
    Expanding the received subspace $\calR$ will, with high probability, lead to a situation in which, upon reduction, the transmitted codeword is recovered.
    \begin{figure}[H]
        \centering
        \begin{tikzpicture}[x=1cm,y=1cm]

        \pgfmathsetmacro{\H}{2}        
        \pgfmathsetmacro{\barw}{0.6}   
        \pgfmathsetmacro{\gap}{0.25}   


        \def\levels{0.4, 0.5,0.90,0.60, 1.0 , 0.55 ,0.63, 0.71}

        \newcount\n
        \n=0
        \foreach \h in \levels { \advance\n by 1 }
        \pgfmathsetmacro{\nbars}{\the\n}
        \pgfmathsetmacro{\W}{\nbars*(\barw+\gap)-\gap}

        \foreach[count=\i] \h in \levels {
        \pgfmathsetmacro{\xL}{(\i-1)*(\barw+\gap)}   
        \pgfmathsetmacro{\xR}{\xL+\barw}             
        \pgfmathsetmacro{\yT}{\h*\H}                 

     \ifnumcomp{\i}{=}{5}{
   \def\thiscolor{green!60}
   \def\thiscolor{blue!75}
    }{
   \ifnumcomp{\i}{=}{3}{
      \def\thiscolor{red!35}
   }{
      \def\thiscolor{red!35}
   }
}

   \node[below] at (3.75,0) {$a \color{blue} \Fqk$};

    \fill[\thiscolor] (\xL,0) rectangle (\xR,\yT);

    \draw[thick] (0,0) -- (7,0);                 
    \draw[thick] (0,\H) -- (7,\H) node[above]{$k$}; 
  }
\end{tikzpicture}
        \caption{Successful expansion: only the space $\calU$ is expanded enough to fill up the correct element of the spread.}
        \label{fig: Successful Expansion}
    \end{figure}
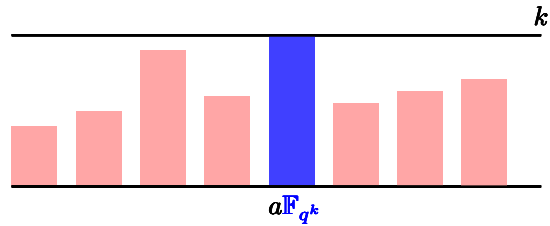
    However, when there exists $b \in \Fqn^*$ such that $b\Fqk\cap a\Fqk=\{0\}$ and $\dim(\calR \cap b \Fqk)\geq k-d$, the expansion $\fexp$ may produce an additional $\Fqk$-linear subspace within $\fexp(\calR)$, which causes a decoding failure.
    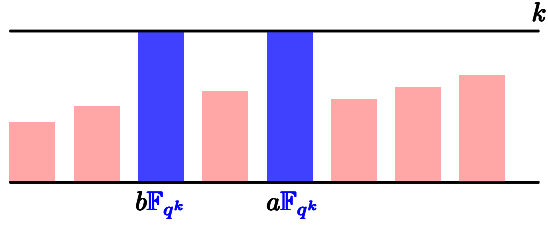
\begin{figure}[H]
        \centering
        \begin{tikzpicture}[x=1cm,y=1cm]

  \pgfmathsetmacro{\H}{2}        
  \pgfmathsetmacro{\barw}{0.6}   
  \pgfmathsetmacro{\gap}{0.25}   

  \def\levels{0.4, 0.5,1.0 ,0.60, 1.0 , 0.55 ,0.63, 0.71}

  \newcount\n
  \n=0
  \foreach \h in \levels { \advance\n by 1 }
  \pgfmathsetmacro{\nbars}{\the\n}
  \pgfmathsetmacro{\W}{\nbars*(\barw+\gap)-\gap}

  \foreach[count=\i] \h in \levels {
    \pgfmathsetmacro{\xL}{(\i-1)*(\barw+\gap)}   
    \pgfmathsetmacro{\xR}{\xL+\barw}             
    \pgfmathsetmacro{\yT}{\h*\H}                 

     \ifnumcomp{\i}{=}{5}{
   \def\thiscolor{green!60}
   \def\thiscolor{blue!75}
    }{
   \ifnumcomp{\i}{=}{3}{
      \def\thiscolor{red!60}
      \def\thiscolor{blue!75}
   }{
      \def\thiscolor{red!35}
   }
}

   \node[below] at (3.75,0) {$a \color{blue} \Fqk$};
   \node[below] at (2,0) {$b \color{blue} \Fqk$};

    \fill[\thiscolor] (\xL,0) rectangle (\xR,\yT);

    \draw[thick] (0,0) -- (7,0);                 
    \draw[thick] (0,\H) -- (7,\H) node[above]{$k$}; 
  }
\end{tikzpicture}
    \caption{Unsuccessful expansion: more than one element of the spread has been filled up by the expansion, we can at best have a list decoding.}
        \label{fig:Expansion}
    \end{figure}

\subsection{Proof of Corollary \ref{cor: expected dimension of the expansion}}
\label{sec: proof of corollary 1}
Following the proofs of \cite[Lemma III.2 and Proposition III.3]{Aragon19}, we provide a complete proof for the case of interest under the assumption $s \dim(\calU)<k$.
\begin{proof}
    Let $\calX_0\coloneqq\{0\}$ and $\calX_i\coloneqq\Span{\Fq}{x_1,\ldots,x_i}$ for all $i\in\{1,\ldots,s\}$, then $\calX_i=\calX_{i-1}+\Span{\Fq}{x_i}$. In order to get $\dim(\calU.\calX)=s \dim(\calU)$, we want that for all $i\in\{1,\ldots,s\}$, $\dim(\calU.\calX_i)=(i-1)\dim(\calU)+\dim(\calU)$. First of all we investigate the probability that $\dim(\calU.\calX_i)<(i-1)\dim(\calU)+\dim(\calU)$ assuming $\dim(\calX_{i-1}.\calU)=(i-1)\dim(\calU)$. We have that $\dim(\calU.\calX_i)<(i-1)\dim(\calU)+\dim(\calU)$ if and only if the subspace $x_i.\calU$ has a non-zero intersection with $\calX_{i-1}.\calU$.
    \begin{align*}
        \mathrm{P}(\calX_{i-1}.\calU\cap x_i.\calU\neq \{0\})&=\mathrm{P}(\text{There exists }u\in\calU, u\neq 0\colon x_iu\in\calX_{i-1}.\calU)\\
        &\leq \sum_{u\in\calU, u\neq 0}\mathrm{P}(x_iu\in\calX_{i-1}.\calU)\\
        &\leq (\lvert\calU\rvert-1)\frac{q^{(i-1)\dim(\calU)}-1}{q^k-1}\\
        &\leq q^{\dim(\calU)}\frac{q^{(i-1)\dim(\calU)}}{q^k}=\frac{q^{i\dim(\calU)}}{q^k},
    \end{align*}
    since for any fixed $x_i,a\neq 0$ such that $\calU\subseteq a\Fqk$, $x_iu=x_iaw\in x_ia\Fqk$ is uniformly distributed, up the isomorphism $w\in\Fqk\mapsto x_iaw\in\Fqn$, in $\Fqk\setminus\{0\}$. At this point, it is straightforward to see that
    \begin{align*}
        &\mathrm{P}(\text{There exists }i\in\{1,\ldots,s\}\colon\dim(\calU.\calX_i)<(i-1)\dim(\calU)+\dim(\calU))\\
        &\leq \sum_{i=1}^s\frac{q^{i\dim(\calU)}}{q^k}\leq \sum_{i=1}^s\frac{q^{s\dim(\calU)}}{q^k}=s\frac{q^{s\dim(\calU)}}{q^k},
    \end{align*}
    which, by complement, immediately yields
    \[
     \mathbb{P}(\dim(\calU.\calX)=s\dim(\calU))\geq 1-s\frac{q^{s\dim(\calU)}}{q^k}=1-s\frac{q^{s(k-d)}}{q^k}.
    \]
\end{proof}

\subsection{About expansion}
\label{app: Expansion}
Keeping in mind the importance of the expansion step in our decoding algorithm, in this section we will deal with the following generic problem.
\begin{problem}
Let $\mathcal{X} \subseteq \mathbb{F}_{q^k}$ be an $\mathbb{F}_q$-linear subspace of dimension $d \le k$.
For each $a \in \Fqk^*$, define the expansion
$
S_a := \mathcal{X} + a\mathcal{X},
$
and consider the average dimension
\[
S := \frac{1}{q^k - 1} \sum_{a \in \Fqk^*} \dim(S_a).
\]
We ask the following questions:
\begin{itemize}
    \item Does $S$ depend only on the dimension $d$ relative to $k$, or also depend on the specific choice of $\mathcal{X}$?
    \item What can be said about the typical behavior?
\end{itemize}
\end{problem}
The first observation is that for $a \in \Fq$ there is no expansion, another trivial instance is $\calX = \Fqk$ for which the expansion is still the whole set independently from the choice of $a$.
The ``typical" case for $\calX$ small (i.e. $d < k/2$) is that $\calX, a \calX$ have trivial intersection $\{0\}$ and the dimension of $\calX +  a \calX$ is equal $2d.$
We will refer to such expansion as \textit{optimal expansion} as it is impossible to get a higher dimension.
The number of optimal expansions is given by $q^k  - |\calX \calX^{-1}|$ (see Lemma \ref{lm:XX_inv}).
An expansion is non-optimal if $\dim_{\Fq}(\calX \cap a \calX) = t \geq 1$, then $\dim_{\Fq}(\calX + a \calX) = 2d - t.$
The following lemma describes the set of all $a \in \Fqk$ for which the expansion $\fexp(\calX,a) = \calX + a \calX$ is non-optimal and gives an upper bound on the cardinality of this set. Although this is an equivalent result to the one presented in \cite[Proposition 3.4]{gluesing2021distance}, we rediscovered it in an attempt to justify optimal expansions. Therefore we leave the proof that led us to the result, since it will help the reader to better understand when an optimal expansion occurs. In particular, the proof gives insight into how bad expansions, meaning expansions that are neither optimal nor quasi-optimal (lack the optimality by one), can have a positive impact on the number of optimal expansions.
\begin{lemma}\label{lm:XX_inv}
    Let $\calX \subseteq \Fqk$ be an $\Fq$-linear subspace of $\Fq$-dimension $d \leq k$ and let $a \in \Fqk^*$. The spaces $\calX, a \calX$ have no trivial intersection if and only if $a \in \calX^{-1} \calX,$ where $\calX^{-1} = \{ x^{-1} \mid x \in \calX \setminus\{0\}\}$  and $\calX^{-1} \calX = \{x_1^{-1} x_2 \mid x_1,x_2 \in \calX, x_1 \neq 0\}.$ The cardinality of $\calX^{-1}\calX$ is upper bounded by 
    \begin{equation} \label{eq:XX_inv_size_bound}
    |\calX^{-1} \calX| \leq  \frac{(q^{d} - 1)}{q-1}(q^d - q)  + q.
\end{equation}

\end{lemma}
\begin{proof}
    The first statement is immediate.
    Consider $\hat{x} \in \calX \cap a \calX$, there exists $x \in \calX$ such that $\hat{x} = ax$ from which we immediately have $ a = x^{-1} \hat{x} \in \calX^{-1} \calX$.

    To measure the size of $\calX^{-1} \calX$, observe that the set $\calX^{-1} \calX$ is the union of several subspaces of the form  $x^{-1}\calX$.
    In particular we have
    \begin{equation}\label{eq:XX_in_union_decomposition}
    \calX^{-1} \calX = \bigcup_{x \in \calX\setminus\{0\}} x^{-1} \calX.
\end{equation} 
Notice that, for any $x \in \calX \setminus\{ 0\}$ and any $\lambda \in \Fq^*$, we have $x^{-1} \calX = (\lambda x)^{-1}\calX,$ this means that in equation (\ref{eq:XX_in_union_decomposition}), instead of taking the union of all the sets $x^{-1} \calX$ we can consider $x \in \calX \setminus \{0\}$ up to $\Fq$-scalar multiplication.
More formally, define the equivalence relation $x \sim_q y \iff x = \lambda y$ for some $\lambda \in \Fq,$ let $[x]_q$ denote the class of $x$ and $[\calX]_q = \{[x]_q \mid x \in \calX \}$ the set of all the classes contained in $\calX$. 
The cardinality of $[\calX]_q$ is exactly $\frac{q^d - 1}{q-1}.$ 
As for any $\lambda x \in [x]_q$ we have $(\lambda x)^{-1} = \lambda^{-1} x^{-1} \sim_q x^{-1},$ then $[x]_q^{-1} := [x^{-1}]_q$ is well defined.
The equation (\ref{eq:XX_in_union_decomposition}) can then be rewritten as
$$
    \calX^{-1} \calX = \bigcup_{[x]_q \in [\calX]_q} [x]_q^{-1} \calX.
$$
Notice that $\Fq \subseteq [x]_q^{-1} \calX$ for each $[x]_q \in \calX$, hence we can rewrite the above union as the union of the sets $[x]_q^{-1} \calX \setminus \Fq$ of cardinality $q^d - q$ and the set $\Fq$.
In this way we obtain the upper bound:
\begin{equation} \label{eq:XX_inv_size_bound_der}
    \left\lvert\calX^{-1} \calX\right\rvert = \left\lvert\bigcup_{x \in \calX} x^{-1} \calX\right\rvert \leq \sum_{[x]_q \in[\calX]_q} \left\lvert \left(x^{-1}  \calX\right) \setminus\Fq\right\rvert + \left\lvert\Fq\right\rvert = \frac{q^{d} - 1}{q-1}(q^d - q)  + q.
\end{equation}
\end{proof}
The upper bound will be tight only if, for any two subspaces $x_1^{-1} \calX, x_2^{-1} \calX$ such that $x_1,x_2 \in \calX \setminus\{ 0\}$ and $[x_1]_q \neq [x_2]_q$ the intersection is always the smallest possible, that is $x_1^{-1} \calX \cap x_2^{-1} \calX = \Fq.$

An example for which the bound in Lemma \ref{lm:XX_inv} is tight is given for $\dim_{\Fq}(\calX) = 1.$ In this case, we have that $\calX$ is  the set $\calX = \{ \lambda x \mid \lambda \in \Fq \}$, the set $\calX^{-1} \calX$ is  $(\lambda_1 x)^{-1} \lambda_2 x = \lambda_1^{-1} \lambda_2 \in \Fq^*,$ that is all the expansions are optimal except choosing $a \in \Fq.$

In Lemma \ref{lm:XX_inv} we describe the set of all the non optimal $a$.
As a consequence we see that for all $\calX$ of dimension at least $2$, besides $a \in \Fq,$ there are always other non-optimal choices of $a.$
Among these non-optimal choices, many could miss the optimality just by $1$, this will not have a big impact on our algorithm.
We would like to characterize the elements $a$ for which $\dim_{\Fq}(\calX + a \calX) \leq 2d - t$ for any given $t.$
For $t=1$ this inequality becomes $\dim_{\Fq}(\calX + a \calX) \leq 2d - 1$ which, thanks to Lemma \ref{lm:XX_inv}, we know is satisfied if and only if $a \in \calX^{-1} \calX$.

In order to characterize the subset for which the expansion misses the optimality by more than $1$ we analyze the case $t=2.$
When the expansion misses the dimension by $2,$ it means that $a$ is sending two $\Fq$-linearly independent $x_1,x_2 \in \calX$ to two linearly independent elements $a x_1, a x_2 \in \calX.$
Let $y_1 = ax_1, y_2 = ax_2 \in \calX$, we have the condition
\begin{equation}\label{eq:a_double}
    a = y_1 x_1^{-1} = y_2 x_2^{-1}.
\end{equation}

This means that $a \in x_1^{-1} \calX \cap x_2^{-1} \calX.$
More in general, if $a$ is such that the expansion will fail by $t$ dimensions, then $a \in \cap_{i=1}^t x_i^{-1} \calX$ where $x_1,\ldots, x_t$ are linearly independent.
The upper bound in Lemma \ref{lm:XX_inv} is tight only when these subspaces intersect only in $\Fq$, a larger intersection including some $a \notin \Fq$ means that the cardinality of the set $\calX^{-1} \calX$ will be strictly smaller than the bound of Lemma \ref{lm:XX_inv}.
As a consequence, if among the non-optimal expansions there are some particularly bad expansions, then the number of optimal expansions will be larger.

This can be clearly seen in the following extreme case.
Let $\calX$ be an intermediate field of dimension $d$ between $\Fq$ and $\Fqk$, the set $ \calX^{-1} \calX = \calX.$
Its size is $q^d,$ which is much smaller than the bound in (\ref{eq:XX_inv_size_bound}).
The possible expansions are either optimal for $a \in \Fqk \setminus \calX$ or terrible.
Indeed, for $a \in \calX$, the dimension of the expansion is the smallest possible as $\calX + a \calX = \calX.$

\par From inequality (\ref{eq:min_distance_decoder_failure}) it follows that, except from the zero insertion case, in order to achieve successful decoding it is needed that $k-d \geq 2$ where $d$ is the number of deletions.
This means that $\dim(\calU) = 2$ represents the worst case for our algorithm.
Let $\calU = \Span{\Fq}{u_1, u_2}$ and choose the vector $\bm a = (a_1,\ldots,a_s)$ to expand it. 
We have $\fexp(\calU,\bm a) = \calU.\calA$ where $\calA = \Span{\Fq}{a_1, \ldots, a_s}$.
As $\dim(u_1^{-1}\calU.\calA) = \dim(\calU.\calA)$ we can apply Lemma \ref{lm:XX_inv} to $\Span{\Fq}{1, u}.\calA$  where $u = u_1^{-1} u_2$ concluding that $\dim(\fexp(\calU,\bm a)) = 2s$ if $u \notin \calA \calA^{-1}$.
The likelihood of $u \notin \calA \calA^{-1}$ can be determined from (\ref{eq:XX_inv_size_bound_der}).
To be precise we have 
\begin{equation}
\label{eq:precise perfect growth probability}
    \mathrm{P}(u \notin \calA \calA^{-1}) = 1 - \frac{|\calA \calA^{-1}|}{q^k} \geq  1  - q^{-k} \left (\frac{q^{2s}-q^{s+1}-q^s+q^2}{q-1}\right) \approx 1 - q^{2s - k - 1}\text{ as }q\to+\infty.
\end{equation}

\bigskip

\noindent Ermes Franch and Chunlei Li,\\
Department of Informatics,\\ 
University of Bergen, Norway.\\
E-mail: \{ermes.franch, chunlei.li\}@uib.no\\

\medskip
\noindent Angelica Piccirillo,\\
Department of Mathematics,\\ 
Technical University of Munich,\\
TUM School of Computation, Information
and Technology (CIT), Germany.\\ 
E-mail: angelica.piccirillo@tum.de

\end{document}